\documentclass[10pt,journal,compsoc,fleqn]{IEEEtran}
\usepackage{graphicx}
\usepackage{algorithm}
\usepackage{algorithmic}
\usepackage{amsmath}
\usepackage{bm}
\usepackage{indentfirst}
\usepackage{amsthm}
\usepackage{multicol}
\usepackage{multirow}
\usepackage{color}
\usepackage{subfigure}
\usepackage{threeparttable}
\usepackage{varwidth,array,ragged2e}
\usepackage{listings}
\usepackage{amsfonts}
\usepackage[numbers,sort&compress]{natbib}
\usepackage{longtable}
\usepackage{rotating}
\usepackage{caption}
\usepackage[justification=centering,font=small]{caption}
\usepackage{amssymb}
\usepackage{CJK}
\usepackage{diagbox}

\usepackage{palatino}

\usepackage{titlesec}
\usepackage{titletoc}

\usepackage{hyperref}
\usepackage{amsmath,amssymb}
\usepackage{makecell}
\usepackage{setspace}

\ifCLASSINFOpdf
\else
\fi

\begin{document}
\title{SGD$\_$Tucker: A Novel Stochastic Optimization Strategy for Parallel Sparse Tucker Decomposition}
\author{
Hao Li, \emph{Student Member}, \emph{IEEE},
Zixuan Li,
Kenli Li, \emph{Senior Member}, \emph{IEEE},
Jan S. Rellermeyer, \emph{Member}, \emph{IEEE},
Lydia Chen, \emph{Senior Member}, \emph{IEEE},
Keqin Li, \emph{Fellow}, \emph{IEEE}.
\IEEEcompsocitemizethanks
{
\IEEEcompsocthanksitem
Hao Li, Zixuan Li, Kenli Li and Keqin Li are with the College of Computer Science and Electronic Engineering, Hunan University, and National Supercomputing Center in Changsha, Hunan, China, 410082.
\IEEEcompsocthanksitem Corresponding author: Kenli Li.
\protect\\ E-mail:
lihao123@hnu.edu.cn (H.Li-9@tudelft.nl),
zixuanli@hnu.edu.cn,
lkl@hnu.edu.cn,
J.S.Rellermeyer@tudelft.nl,
Y.Chen-10@tudelft.nl,
lik@newpaltz.edu.
\IEEEcompsocthanksitem Hao Li is also with TU Delft, Netherlands.
\IEEEcompsocthanksitem J.S.Rellermeyer, Lydia Chen are with TU Delft, Netherlands.
\IEEEcompsocthanksitem Keqin Li is also with the Department of Computer Science, State University of New York, New Paltz, New York 12561, USA.
}
%
}

\markboth{}%
{Shell \MakeLowercase{\textit{et al.}}: Bare Advanced Demo of IEEEtran.cls for Journals}
\IEEEtitleabstractindextext{%
\begin{abstract}
Sparse Tucker Decomposition (STD) algorithms learn a core tensor and a group of factor matrices to obtain an optimal low-rank representation feature for the
\underline{H}igh-\underline{O}rder, \underline{H}igh-\underline{D}imension, and \underline{S}parse \underline{T}ensor (HOHDST).
However,
existing STD algorithms face the problem of intermediate variables explosion which results from the fact that the formation of those variables, i.e., matrices Khatri-Rao product, Kronecker product, and
matrix-matrix multiplication, follows the whole elements in sparse tensor.
The above problems prevent deep fusion of efficient computation and big data platforms.
To overcome the bottleneck,
a novel stochastic optimization strategy (SGD$\_$Tucker) is proposed for STD
which can automatically divide the high-dimension intermediate variables
into small batches of intermediate matrices.
Specifically, SGD$\_$Tucker only follows the randomly selected small samples rather than the whole elements, while maintaining the overall accuracy and convergence rate.
In practice,
SGD$\_$Tucker features the two distinct advancements over the state of the art.
First, SGD$\_$Tucker can prune the communication overhead for the core tensor in distributed settings.
Second, the low data-dependence of SGD$\_$Tucker enables fine-grained parallelization,
which makes SGD$\_$Tucker obtaining lower computational overheads with the same accuracy.
Experimental results show that SGD$\_$Tucker runs at least 2$X$ faster than the state of the art.
\renewcommand{\raggedright}{\leftskip=0pt \rightskip=0pt plus 0cm}
\raggedright
\end{abstract}
}
\maketitle

\IEEEdisplaynontitleabstractindextext

\IEEEpeerreviewmaketitle
\section{Introduction} \label{section1}
\renewcommand{\raggedright}{\leftskip=0pt \rightskip=0pt plus 0cm}
\raggedright
\IEEEPARstart{T}{e}nsors are a widely used data representation style for interaction data in the Machine Learning (ML) application community \cite{ex102}, e.g,
in Recommendation Systems \cite{ex108},
Quality of Service (QoS) \cite{ex109},
Network Flow \cite{ex105},
Cyber-Physical-Social (CPS) \cite{ex211}, or
Social Networks \cite{ex107}.
In addition to applications in which the data is naturally represented in the form of tensors,
another common used case is the fusion in multi-view or multi-modality problems~\cite{ex104}.
Here, during the learning process,
each modality corresponds to a feature and the feature alignment involves fusion.
Tensors are a common form of feature fusion for multi-modal learning \cite{ex104, ex111, ex112, ex115}.
Unfortunately, tensors can be difficult to process in practice.
For instance, an $N$-order tensor comprises of the interaction relationship between $N$ kinds of attributes and
if each attribute has millions of items,
this results in a substantially large size of data \cite{ex110}.
As a remedy,
\emph{dimensionality reduction} can be used to represent the original state using much fewer parameters \cite{ex117}.

Specifically in the ML community,
Tensor Factorization (TF), as a classic dimensionality reduction technique, plays a key role for low-rank representation.
Xu et al., \cite{ex202} proposed a Spatio-temporal multi-task learning model via TF and in this work,
tensor data is of $5$-order, i.e., weather, traffic volume, crime rate, disease incidents, and time.
Meanwhile, this model made predictions through the time-order for the multi-task in weather, traffic volume, crime rate, and disease incidents orders and the relationship construction between those orders is via TF.
In the community of Natural Language Processing (NLP),
Liu et al., \cite{ex204} organized a mass of texts into a tensor and each slice is modeled as a sparse symmetric matrix.
Furthermore, the tensor representation is a widely-used form for Convolutional Neural Networks (CNNs), e.g., in the popular TensorFlow~\cite{ex127} framework, and
Kossaifi et al., \cite{ex203} took tensorial parametrization of a CNNs and pruned the parametrization by Tensor Network Decomposition (TND).
Meanwhile, Ju, et al., \cite{ex205} pruned and then accelerated the Restricted Boltzmann Machine (RBM) coupling with TF.
In the Computer Vision (CV) community,
Wang et al., \cite{ex210} modeled
various factors, i.e., pose and illumination, as an unified tensor and make disentangled representation by adversarial autoencoder via TF.
Zhang et al., \cite{ex201} constructed multi subspaces of multi-view data and then abstract factor matrices via TF for
the unified latent of each view.

\underline{H}igh-\underline{O}rder, \underline{H}igh-\underline{D}imension, and \underline{S}parse \underline{T}ensor (HOHDST) is a common situation in the big-data processing and ML application community \cite{ex113, ex114}.
Dimensionality reduction can also be used to find the low-rank space of HOHDST in ML applications \cite{ex119},
which can help to make prediction from existing data.
Therefore, it is non-trivial to learn the pattern from the existed information in a HOHDST and then make the corresponding prediction.
Sparse Tucker Decomposition (STD) is one of the most popular TF models for HOHDST, which can find the $N$-coordinate systems and those systems are tangled by a core tensor between each other \cite{ex121}.
Liu et al., \cite{ex239} proposed to accomplish the visual tensor completion via STD.
The decomposition process of STD involves the entanglement of $N$-factor matrices and core tensor and
the algorithms follow one of the following two directions:
1) search for optimal orthogonal coordinate system of factor matrices, e.g., High-order Orthogonal Iteration (HOOI)~\cite{ex214};
2) designing optimization solving algorithms~\cite{ex123}.
HOOI is a common solution for STD~\cite{ex122} and
able to find the $N$ orthogonal coordinate systems which are similar to Singular Value Decomposition (SVD), but requires frequent intermediate variables of Khatri-Rao and Kronecker products \cite{ex120}.

An interesting topic is that
stochastic optimization~\cite{ex221}, e.g., Count Sketch and Singleshot~\cite{ex222, ex223}, etc. can alleviate this bottleneck to a certain extent, depending on the size of dataset.
However, those methods depend on the count sketch matrix, which cannot be easily implemented in a distributed environment and is notoriously difficult to parallelize.
The Stochastic Gradient Descent (SGD) method can approximate the gradient from randomly selected subset and
it forms the basis of the state of art methods, e.g.,
variance SGD \cite{ex225},
average SGD \cite{ex229},
Stochastic Recursive Gradient \cite{ex230}, and
momentum SGD \cite{ex226}.
SGD is adopted to approximate the eventual optimal solver with lower space and computational complexities;
meanwhile, the low data-dependence makes the SGD method amenable to parallelization \cite{ex125}.
The idea of construction for the computational graph of the mainstream platforms, e.g., Tensorflow~\cite{ex127} and Pytorch~\cite{ex128}, is
based on the SGD~\cite{ex240} and
practitioners have already demonstrated its powerful capability on large-scale optimization problems.
The computational process of SGD only needs a batch of training samples rather than the full set which gives
the ML algorithm the low-dependence between each data block and low communication overhead \cite{ex227}.

There are three challenges to process HOHDST in a fast and accurate way:
1) how to define a suitable optimization function to find the optimal factor matrices and core tensor?
2) how to find an appropriate optimization strategy in a low-overhead way and then reduce the entanglement of the factor matrices with core tensor which
may produce massive intermediate matrices?
3) how to parallelize STD algorithms and make distributed computation with low communication cost?
In order to solve these problems, we present the main contributions of this work which are listed as follows:
\begin{enumerate}
  \item A novel optimization objective for STD is presented.
  This proposed objective function not only has a low number of parameters via coupling the Kruskal product (Section~\ref{Section31}) but also is approximated as a convex function;
  \item A novel stochastic optimization strategy is proposed for STD, SGD$\_$Tucker,
which can automatically divide the high-dimension intermediate variables
into small batches of intermediate matrices that only follows the index of the randomly selected small samples;
meanwhile, the overall accuracy and convergence are kept (Section \ref{Section33});
  \item The low data-dependence of SGD$\_$Tucker creates opportunities for fine-grained parallelization,
which makes SGD$\_$Tucker obtaining lower computational overhead with the same accuracy.
Meanwhile, SGD$\_$Tucker does not rely on the specific compressive structure of a sparse tensor (Section \ref{Section35}).
\end{enumerate}

To our best knowledge,
SGD$\_$Tucker is the first work that can divide the high-dimension intermediate matrices
of STD into small batches of intermediate variables, a critical step for fine-grained parallelization with low communication overhead.
In this work,
the related work is presented in Section~\ref{Sectionrela}.
The notations and preliminaries are introduced in Section~\ref{Section2}.
The SGD$\_$Tucker model as well as parallel and communication overhead on distributed environment for STD are showed in Section~\ref{Section3}.
Experimental results are illustrated in Section~\ref{Section4}.

\section{Related Studies} \label{Sectionrela}
For HOHDST, there are many studies to accelerate STD on the state of the art parallel and distributed platforms, i.e., OpenMP, MPI, CUDA, Hadoop, Spark, and OpenCL.
Ge et al.,  \cite{ex232} proposed distributed CANDECOMP/PARAFAC Decomposition (CPD) which is a special STD for HOHDST.
Shaden et al., \cite{ex233} used a Compressed Sparse Tensors (CSF) structure which can optimize the access efficiency for HOHDST.
Tensor-Time-Matrix-chain (TTMc)\cite{ex131} is a key part for \textcolor[rgb]{0.00,0.00,1.00}{Tucker Decomposition (TD)} and TTMc is
a data intensive task.
Ma et al., \cite{ex131} optimized the TTMc operation on GPU which
can take advantage of intensive and partitioned computational resource of GPU, i.e.,
a warp threads (32) are automatically synchronized and this mechanism is apt to matrices block-block multiplication.
Non-negative Tucker Decomposition (NTD) can extract the non-negative latent factor of a HOHDST,
which is widely used in ML community.
However, NTD need to construct the numerator and denominator and both of them involve TTMc.
Chakaravarthy et al., \cite{ex130} designed a mechanism which can
divide the TTMc into multiple independent blocks and then those tasks are allocated to distributed nodes.

HOOI is a common used TF algorithm which comprises of a series of TTMc matrices and SVD for the TTMc.
\cite{ex132} focused on dividing the computational task of TTMc into a list of independent parts and
a distributed HOOI for HOHDST is presented in \cite{ex235}.
However, the intermediate cache memory for TTMc of HOOI increased explosively.
Shaden et al., \cite{ex233} presented a parallel HOOI algorithm and
this work used a Compressed Sparse Tensors (CSF) structure which can improve the access efficiency and save the memory for HOHDST.
Oh and Park \cite{ex234, ex237} presented a parallelization strategy of ALS and CD for STD on OpenMP.
A heterogeneous OpenCL parallel version of ALS for STD is proposed on \cite{ex236}.
However, the above works are still focused on the naive algorithm parallelization which
cannot solve the problem of fine-grained parallelization.

\section{Notation and Preliminaries} \label{Section2}
The main notations include the tensors, matrices and vectors, along with their basic elements and operations (in Table \ref{table21}).
Fig.~\ref{fig202} illustrates the details of \textcolor[rgb]{0.00,0.00,1.00}{TD} including the tanglement of the core tensor $\mathcal{G}$ with the $N$ factor matrices $\textbf{A}^{(n)}$, $n$ $\in$ $\{N\}$.
The data styles for TF include sparse and dense tensors and STD is devoted to HOHDST.
Here,
basic definitions for STD and models are rendered in Section \ref{section21}.
Finally, the STD process for the HOHDST is illustrated in Section \ref{section23}.

\begin{table}[!htbp]
\setlength{\abovedisplayskip}{0pt}
\setlength{\belowdisplayskip}{0pt}
\renewcommand{\arraystretch}{1.0}
\caption{Table of symbols.}
\centering
\label{table21}
\tabcolsep1pt
\begin{tabular}{cl}
\hline
\hline
 Symbol\ \ &\ \ Definition\\
\hline
$I_{n}$ \ \ &\ \ The size of row in the $n$th factor matrix; \\
$J_{n}$ \ \ &\ \ The size of column in the $n$th factor matrix; \\
$\mathcal{X}$ \ \ &\ \ Input $N$ order tensor $\in$ $\mathbb{R}^{I_{1}\times I_{2}\times\cdots \times I_{N}}_{+}$; \\
$x_{i_{1},i_{2},\cdots,i_{n}}$\ \ &\ \ $i_{1},i_{2},\cdots,i_{n}$th element of tensor $\mathcal{X}$;\\
$\mathcal{G}$ \ \ &\ \ Core $N$ order tensor $\in$ $\mathbb{R}^{J_{1}\times J_{2}\times\cdots \times J_{N}}$;\\
$\textbf{X}$\ \ &\ \ Input matrix $\in$ $\mathbb{R}^{I_{1}\times I_{2}}_{+}$;\\
$\textbf{X}^{(n)}$\ \ &\ \ $n$th unfolding matrix for tensor $\mathcal{X}$;\\
Vec$_{n}$($\bm{\mathcal{X}})$\ \ &\ \ $n$th vectorization of a tensor $\mathcal{X}$;\\
$\Omega$\ \ &\ \ Index $(i_{1},\cdots,i_{n},\cdots,i_{N})$ of a tensor $\mathcal{X}$;\\
$\Omega^{(n)}_{M}$\ \ &\ \ Index $(i_{n},j)$ of $n$th unfolding matrix $\textbf{X}^{(n)}$;\\
$(\Omega^{(n)}_{M})_{i}$\ \ &\ \ Column index set in $i$th row of $\Omega^{(n)}_{M}$;\\
$(\Omega^{(n)}_{M})^{j}$\ \ &\ \ Row index set in $j$th column of $\Omega^{(n)}_{M}$;\\
$\Omega^{(n)}_{V}$\ \ &\ \ Index $i$ of $n$th unfolding vector Vec$_{n}$($\bm{\mathcal{X}}$);\\
$\{N\}$           \ \ &\ \ The ordered set $\{1,2,\cdots,N-1,N\}$;\\
$\textbf{A}^{(n)}$\ \ &\ \ $n$th feature matrix $\in$ $\mathbb{R}^{I_{n}\times J_{n}}$;\\
$a_{i_{n},:}^{(n)}$\ \ &\ \ $i_{n}$th row vector $\in$ $\mathbb{R}^{K_{n}}$ of $\textbf{A}^{(n)}$;\\
$a_{:,j}^{(n)}$\ \ &\ \ $j$th column vector $\in$ $\mathbb{R}^{K_{n}}$ of $\textbf{A}^{(n)}$;\\
$a_{i_{n},k_{n}}^{(n)}$\ \ &\ \ $k_{n}$th element of feature vector $a_{i_{n},:}^{(n)}$;\\
$\cdot \big/ (-,/)$\ \  &\ \ Element-wise multiplication$\big/$ division;\\
$\circ$\ \  &\ \ Outer production of vectors;\\
$\odot$\ \  &\ \ Khatri-Rao (columnwise Kronecker) product;\\
$\times$\ \  &\ \ Matrix product;\\
$\times_{(n)}$\ \  &\ \ $n$-Mode Tensor-Matrix product;\\
$\otimes$\ \  &\ \ Kronecker product.\\
 \hline
 \hline
\end{tabular}
\end{table}

\subsection{Basic Definitions} \label{section21}
\newtheorem{definition}{Definition}
\begin{definition}[Tensor Unfolding (Matricization)]
$n$th tensor unfolding (Matricization) refers to that
a low order matrix
$\textbf{X}^{(n)}\in\mathbb{R}^{I_{n}\times I_{1}\cdots I_{n-1}\cdot I_{n+1}\cdot\cdots\cdot I_{N}}$ stores all information of a tensor
$\mathcal{X}\in\mathbb{R}^{I_{1}\times I_{2}\times\cdots I_{n}\cdots \times I_{N}}$ and
the matrix element $x^{(n)}_{i_{n},j}$ of $\textbf{X}^{(n)}$ at the position $j=1+\sum_{k=1,n\neq k}^{N}\left[(i_{k}-1)\mathop{\prod}_{m=1,m\neq n}^{k-1}I_{m}\right]$ contains the tensor element $x_{i_{1},i_{2},\cdots,i_{n}}$ of a tensor $\bm{\mathcal{X}}$.
\end{definition}

\begin{definition}[Vectorization of a tensor $\bm{\mathcal{X}}$]
$n$th tensor vectorization refers to that a vector $x^{(n)}$ (Vec$_{n}$($\bm{\mathcal{X}}$) and Vec($\textbf{X}^{(n)}$)) stores all elements in the $n$th matricization $\textbf{X}^{(n)}$ of a tensor $\mathcal{X}$ and $x^{(n)}_{k}$ $=$ $\textbf{X}^{(n)}_{i,j}$, where $k=(j-1)I_{n}+i$.
\end{definition}

\begin{definition}[Tensor Approximation]
A $N$-order tensor $\bm{\mathcal{X}}$ $\in$ $\mathbb{R}^{I_{1}\times\cdots\times I_{N}}$ can be approximated by $\widehat{\bm{\mathcal{X}}}$ $\in$ $\mathbb{R}^{I_{1}\times\cdots\times I_{N}}$, as well as
a $N$-order residual or noisy tensor $\bm{\mathcal{E}}$ $\in$ $\mathbb{R}^{I_{1} \times\cdots\times I_{N}}$.
The low-rank approximation problem is defined as $\bm{\mathcal{X}}=\widehat{\bm{\mathcal{X}}}+\bm{\mathcal{E}}$,
where $\widehat{\bm{\mathcal{X}}}$ is denoted by a low-rank tensor.
\end{definition}

\begin{figure}
\centering
\includegraphics[width=3.5in,height=2.0in]{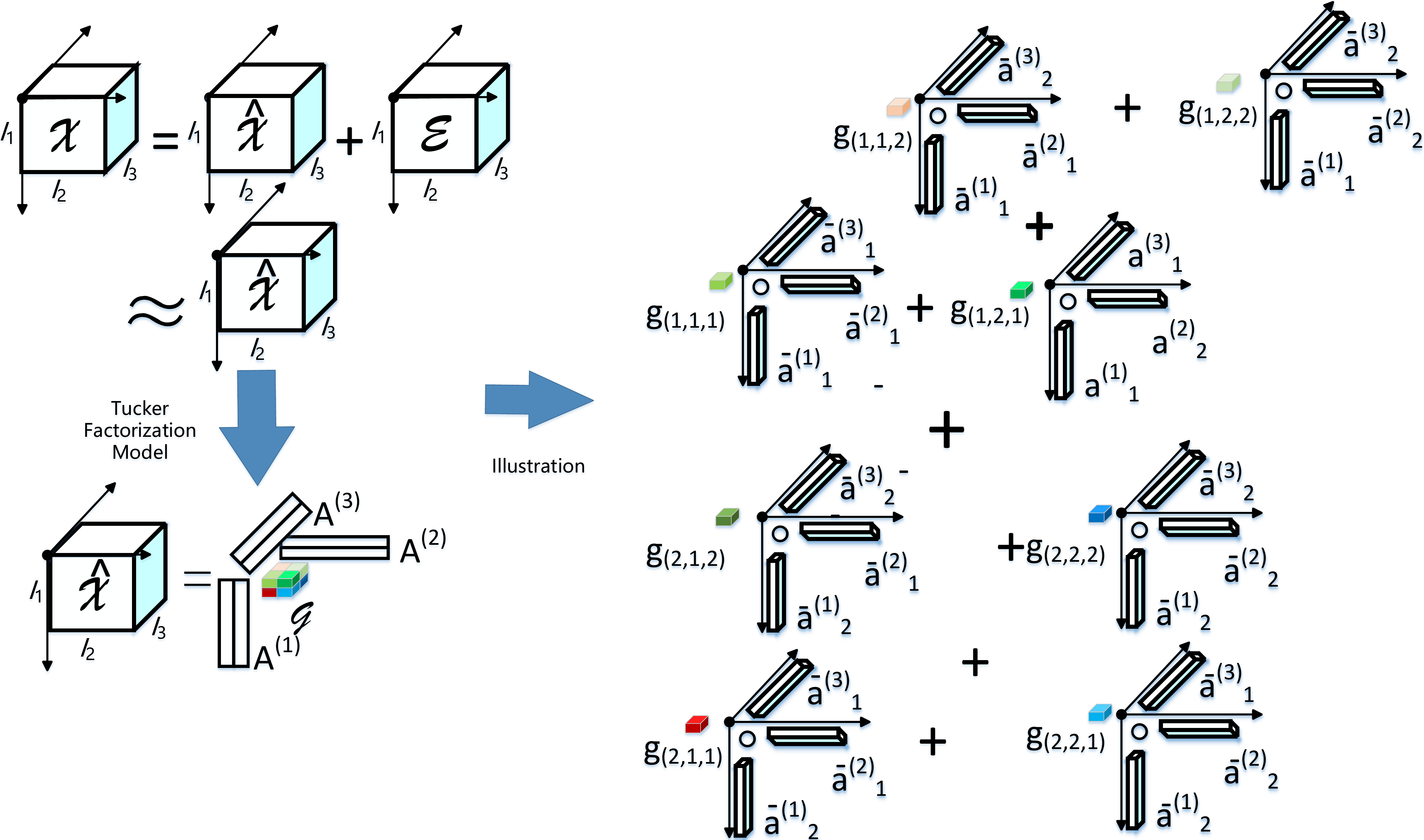}
\caption{Illustration of \textcolor[rgb]{0.00,0.00,1.00}{TD}.}
\label{fig202}
\end{figure}

\begin{definition}[$n$-Mode Tensor-Matrix product]
$n$-Mode Tensor-Matrix product is an operation which can reflect coordinate projection of a tensor $\bm{\mathcal{X}}$ $\in$ $\mathbb{R}^{I_{1}\times\cdots\times I_{N}}$ with projection matrix $\textbf{U}$ $\in$ $\mathbb{R}^{I_{n}\times J_{n}}$ into a tensor $(\bm{\mathcal{X}}\times_{(n)} \textbf{U})$ $\in$ $\mathbb{R}^{I_{1}\times\cdots \times I_{n-1}\times J_{n}\times\cdots  I_{N}}$ where $(\bm{\mathcal{X}}\times_{(n)} \textbf{U})_{i_{1}\times\cdots \times i_{n-1}\times j_{n}\times i_{n+1}\times\cdots\times  i_{N}}$ $=$ $\sum\limits_{i_{n}=1}^{I_{n}}$ $x_{i_{1}\times\cdots \times i_{n}\times \cdots\times i_{N}}u_{j_{n},i_{n}}$.
\end{definition}

\begin{definition}[$R$ Kruskal Product]
For an $N$-order tensor $\widehat{\bm{\mathcal{X}}}$ $\in$ $\mathbb{R}^{I_{1}\times\cdots \times I_{N}}$,
the $R$ Kruskal Product of $\widehat{\bm{\mathcal{X}}}$ is given by $R$ Kruskal product as:
$\widehat{\bm{\mathcal{X}}}=\sum_{r=1}^{R} a^{(1)}_{:,r}\circ\cdots\circ a^{(n)}_{:,r}\circ\cdots\circ a^{(N)}_{:,r}$.
\end{definition}

\begin{definition}[Sparse Tucker Decomposition (STD)]
For a $N$-order sparse tensor $\bm{\mathcal{X}}$ $\in$ $\mathbb{R}^{I_{1}\times\cdots \times I_{N}}$,
the STD of the optimal approximated $\widehat{\bm{\mathcal{X}}}$ is given by
$\widehat{\bm{\mathcal{X}}}=\mathcal{G}\times_{(1)}\textbf{A}^{(1)}\times_{(2)}\cdots\times_{(n)}\textbf{A}^{(n)}\times_{(n+1)}\cdots\times_{(N)}\textbf{A}^{(N)},
$
where $\mathcal{G}$ is the core tensor and
$\textbf{A}^{(n)}$, $n \in \{N\}$ are the low-rank factor matrices.
The rank of TF for a tensor is $[J_{1},\cdots,J_{n},\cdots,J_{N}]$.
The determination process for the core tensor $\mathcal{G}$ and factor matrices $\textbf{A}^{(n)}$, $n \in \{N\}$ follows the sparsity model of the sparse tensor $\bm{\mathcal{X}}$.
\end{definition}

\textcolor[rgb]{0.00,0.00,1.00}{In the convex optimization community, the literature \cite{ex243} gives the definition of Lipschitz-continuity  with constant $L$ and strong convexity with constant $\mu$.}

\begin{definition}[$L$-Lipschitz continuity]
\textcolor[rgb]{0.00,0.00,1.00}{A continuously differentiable function $f(\textbf{x})$ is called $L$-smooth on $\mathbb{R}^{r}$ if the gradient $\nabla f(\textbf{x})$ is $L$-Lipschitz continuous for any $\textbf{x}$, $\textbf{y}$ $\in$ $\mathbb{R}^{r}$, that is
$\|$ $\nabla f(\textbf{x})$ $-$ $\nabla f(\textbf{y})$ $\|_{2}$ $\leq$ $L$ $\|$ $\textbf{x}$ $-$ $\textbf{y}$ $\|_{2}$,
where $\|\bullet\|_{2}$ is $L_{2}$-norm $\|\textbf{x}\|_{2}$ $=$ $(\mathbb{\sum}_{k=1}^{r}x_{k}^{2})^{1/2}$ for a vector $\textbf{x}$.}
\end{definition}

\begin{definition}[$\mu$-Convex]
\textcolor[rgb]{0.00,0.00,1.00}{A continuously differentiable function $f(\textbf{x})$ is called strongly-convex on $\mathbb{R}^{r}$ if there exists a constant $\mu$ $>$ $0$ for any $\textbf{x}$, $\textbf{y}$ $\in$ $\mathbb{R}^{r}$, that is
$f(\textbf{x})$ $\geq$ $f(\textbf{y})$ $+$ $\nabla$ $f(\textbf{y})$ $(\textbf{x}-\textbf{y})^{T}$ $+$ $\frac{1}{2}\mu\|\textbf{x}-\textbf{y}\|_{2}^{2}$.}
\end{definition}

Due to limited spaces, we provide the description of basic optimization as  the supplementary material.

\subsection{Problems for Large-scale Optimization Problems} \label{section23}
Many ML tasks are transformed into the solvent of optimization problem~\cite{ex225, ex226, ex229, ex230} and the basis optimization problem is organized as:
 \begin{equation}\label{Original}
  \begin{aligned}
  \mathop{\arg\min}_{w \in \mathbb{R}^{R}} f(w)&=\underbrace{L\bigg(w\bigg|y_{i}, x_{i}, w\bigg)}_{Loss\ Function}+ \underbrace{\lambda_{w} R(w)}_{Regularization\ Item}\\
  &=\sum\limits_{i=1}^{N} L_{i}\bigg(w\bigg|y_{i}, x_{i}, w\bigg)+ \lambda_{w}R_{i}(w),
   \end{aligned}
\end{equation}
where $y_{i}$ $\in$ $\mathbb{R}^{1}$, $x_{i}$ $\in$ $\mathbb{R}^{R}$, $i\in \{N\}$, $w \in \mathbb{R}^{R}$ and the original optimization model needs gradient which should select all the samples $\{x_{i}|i \in \{N\}\}$ from the dataset $\Omega$ and the GD is presented as:
 \begin{equation}\label{GD}
  \begin{aligned}
  w&\leftarrow w- \gamma \frac{ \partial f_{\Omega}(w) }{ \partial w }\\
   &= w- \gamma \frac{1}{N}\sum\limits_{i=1}^{N} \frac{ \partial \bigg(L_{i}(w) + \lambda_{w}R_{i}(w) \bigg) }{ \partial w }.
   \end{aligned}
\end{equation}
The second-order solution, i.e., ALS and CD, etc, are deduced from the construction of GD from the whole dataset.
In large-scale optimization scenarios,
SGD is a common strategy \cite{ex225, ex226, ex229, ex230} and promises to obtain the optimal accuracy via a certain number of training epoches \cite{ex225, ex226, ex229, ex230}.
An $M$ entries set $\Psi$ is randomly selected from the set $\Omega$, and the SGD is presented as:
 \begin{equation}\label{SGD}
  \begin{aligned}
  w&\leftarrow w- \gamma \frac{ \partial f_{\Psi}(w) }{ \partial w }\\
   &\approx w- \gamma \frac{1}{M}\sum\limits_{i\in \Psi} \frac{ \partial \bigg(L_{i}(w) + \lambda_{w}R_{i}(w) \bigg) }{ \partial w }.
   \end{aligned}
\end{equation}
Equ. (\ref{SGD}) is an average SGD \cite{ex229}, and
the averaged SG can be applied to build the basic tool of the modern stochastic optimization strategies, e.g.,
Stochastic Recursive Gradient~\cite{ex230},
variance SGD~\cite{ex225}, or
momentum SGD~\cite{ex226},
which can retain robustness and fast convergence.
The optimization function can be packaged in the form of
$SGD\bigg(M, \lambda, \gamma, w, \frac{ \partial f_{\Psi}(w) }{ \partial w }\bigg)$.

\begin{figure}
\centering
\includegraphics[width=3.5in,height=3.0in]{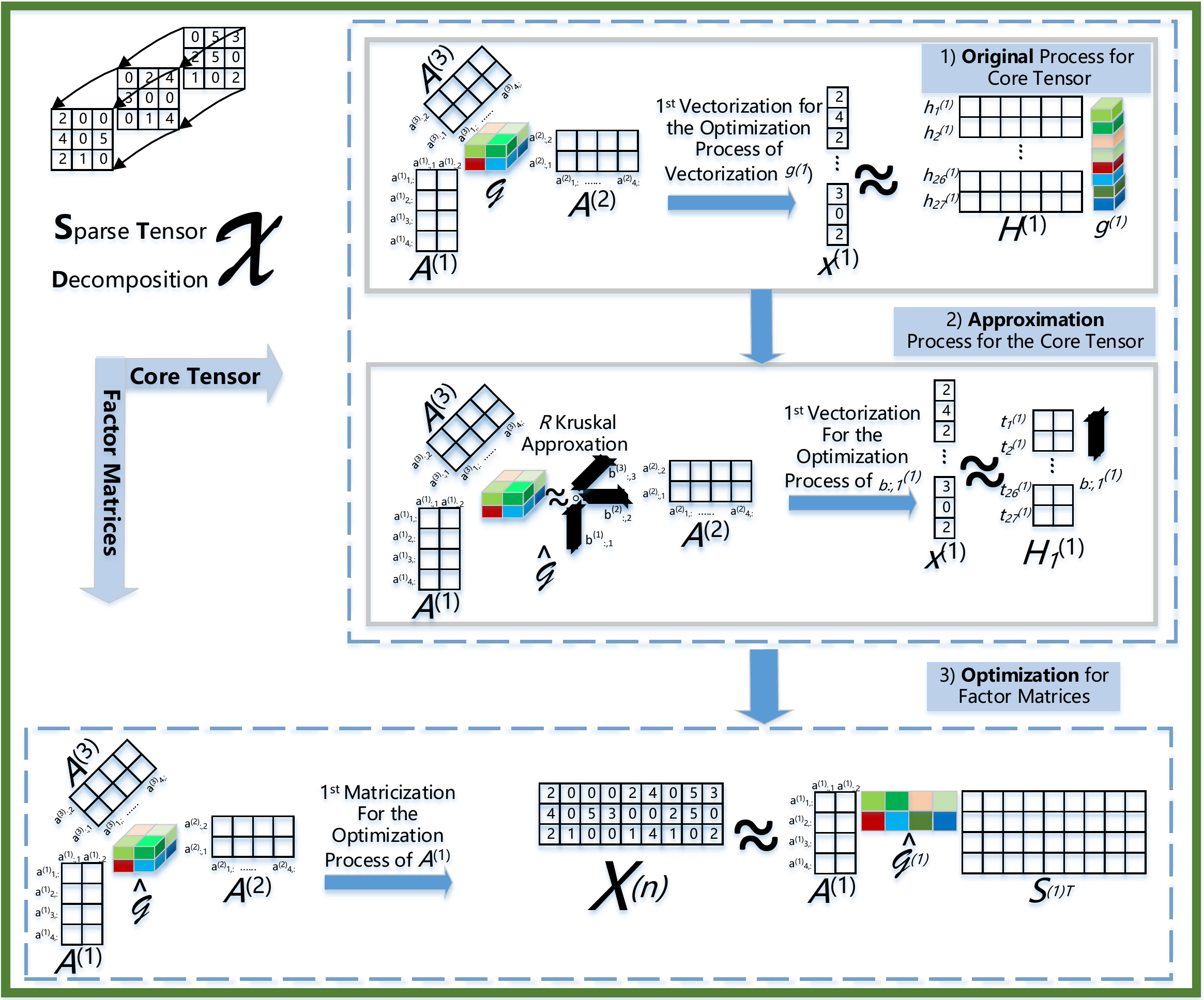}
\caption{Illustration of Optimization Process for SGD$\_$Tucker:
1) Original problem for core tensor,
2) Kruskal product for approximating core tensor,
3) Optimization process for factor matrices.}
\label{fig301}
\end{figure}

\section{SGD$\_$Tucker} \label{Section3}
$\textbf{Overview}$: In this section,
SGD$\_$Tucker is proposed to decompose the optimization objectives which involves frequent operations of intermediate matrices
into a problem which only needs the manipulations of small batches of intermediate matrices.
Fig.~\ref{fig301} illustrates the framework for our work.
As Fig.~\ref{fig301} shows, SGD$\_$Tucker comprises of an approximation process of the core tensor and an optimization for factor matrices.
Table~\ref{table31} records the intermediate variables which follow the problem deduction process.
Section~\ref{Section31} presents the deduction process for the core tensor (Lines 1-16, Algorithm \ref{alg31}),
and
Section~\ref{Section32} presents the proceeding process for the factor matrices of SGD$\_$Tucker (Lines 17-26, Algorithm \ref{alg31}).
Section~\ref{Section33} shows the stochastic optimization process for the proposed novel optimization objectives, which
illustrates the details of the process that how can the SGD$\_$Tucker divide the high-dimension intermediate matrices
$\bigg\{\textbf{H}^{(n)}, \textbf{S}^{(n)}, \textbf{E}^{(n)}\big|n \in \{N\}\bigg\}$ into a
small batches of intermediate matrices.
Section~\ref{Section35} shows the parallel and distributed strategies for SGD$\_$Tucker.
Finally, section~\ref{Section34} illustrates the analysis of space and time complexities.
Because the optimization process for approximating the core tensor and factor matrices are non-convex.
We fix the core tensor and optimize a factor matrix and fix factor matrices to optimize Kruskal product for core tensor.
This alternative minimization strategy is a convex solution \cite{ex234, ex237, ex236, ex241}.

\begin{table}[!htbp]
\setlength{\abovedisplayskip}{0pt}
\setlength{\belowdisplayskip}{0pt}
\renewcommand{\arraystretch}{1.0}
\caption{Table of Intermediate Varibles for SGD$\_$Tucker.}
\centering
\label{table31}
\tabcolsep1pt
\begin{tabular}{lll}
\hline
\hline
 Symbol & Description & Emerging \\
        &             & Equations \\
\hline
$\widehat{\bm{\mathcal{G}}}$ $\in$ $\mathbb{R}^{J_{1}\times J_{2}\times\cdots \times J_{N}}_{+}$ & $\mathbb{R}_{core}$ Kruskal                    & Equ. (\ref{CP-Decomposition});\\
                             & product for $\bm{\mathcal{G}}$        & \\
$\textbf{B}^{(n)}\in \mathbb{R}^{J_{n}\times R_{core}}$  & $n$th Kruskal matrix for $\widehat{\bm{\mathcal{G}}}$ & Equ. (\ref{CP-Decomposition});\\
$\widehat{g}^{(n)}$ $\in$ $\mathbb{R}^{\prod\limits_{n=1}^{N}J_{n}}$  & $n$th vectorization & Equ. (\ref{approximated_core_tensor_optimization});\\ 
                                                   & of tensor $\widehat{\bm{\mathcal{G}}}$ &\\
$\textbf{H}^{(n)}$ $\in$ $\mathbb{R}^{\prod\limits_{n=1}^{N}I_{n}\times \prod\limits_{n=1}^{N}J_{n}}$ & Coefficient for $\widehat{g}^{(n)}$                                       & Equ. (\ref{approximated_core_tensor_optimization});\\
$x^{(n)}$ $\in$ $\mathbb{R}^{\prod\limits_{n=1}^{N}I_{n}}$ & $n$th vectorization of     & Equ. (\ref{approximated_core_tensor_optimization});\\
                                                           & Tensor $\mathcal{X}$    & \\
$\textbf{Q}^{(n)}$ $\in$ $\mathbb{R}^{\prod\limits_{r=1, r\neq n}^{N}J_{r}\times R_{core}}$ & Coefficient of $\textbf{B}^{(n)}$ & Equ. (\ref{eventual0_approximated_core_tensor_optimization});\\
                    & for constructing $\widehat{g}^{(n)}$ & \\
$\textbf{U}^{(n)}$ $\in$ $\mathbb{R}^{J_{n}\times J_{n}}$ & Unity matrix & Equ. (\ref{Q_r});\\
$\textbf{O}^{(n)}_{r}$ $\in$ $\mathbb{R}^{\prod\limits_{r=1}^{N}J_{r}\times J_{n}}$  & Coefficient of $b^{(n)}_{:,r}$ & Equ.
(\ref{eventual1_approximated_core_tensor_optimization});\\
                    & for constructing $\widehat{g}^{(n)}$ & \\
$x^{(n)}_{r_{core}}$ $\in$ $\mathbb{R}^{\prod\limits_{r=1}^{N}I_{r}}$ & Intermediate vector of $x^{(n)}$    & Equ. (\ref{eventual_approximated_core_tensor_optimization});\\
$\textbf{H}^{(n)}_{r_{core}}$ $\in$ $\mathbb{R}^{\prod\limits_{n=1}^{N}I_{n}\times J_{n}}$ & Coefficient of $b^{(n)}_{:,r}$ & Equ. (\ref{eventual_approximated_core_tensor_optimization});\\
                    & for constructing $\widehat{g}^{(n)}$ & \\
$\textbf{S}^{(n)}$ $\in$ $\mathbb{R}^{\prod\limits_{k=1,k\neq n}^{N} J_{n} \times \prod\limits_{k=1,k\neq n}^{N}I_{k}}$ & Coefficient of $\textbf{A}^{(n)}\widehat{\textbf{G}}^{(n)}$ & Equ. (\ref{optimization_for_factor});\\
                    & for constructing $\widehat{\bm{X}}^{(n)}$ & \\
$\textbf{E}^{(n)}$ $\in$ $\mathbb{R}^{J_{n} \times \prod\limits_{k=1,k\neq n}^{N}I_{k}}$ & Coefficient of $\textbf{A}^{(n)}$ & Equ. (\ref{Alternative_Optimization_for_Factor_Matrices});\\
                    & for constructing $\widehat{\bm{X}}^{(n)}$ &\\
$x^{(n)}_{\Psi^{(n)}_{V}}$ $\in$ $\mathbb{R}^{|\Psi^{(n)}_{V}|}$ & Partial vector $x^{(n)}$ from set $\Psi$  & Equ. (\ref{SGD_core_tensor});\\
$\textbf{H}_{\Psi^{(n)}_{V},:}^{(n)}$ $\in$ $\mathbb{R}^{|\Psi^{(n)}_{V}| \times \prod\limits_{n=1}^{N}J_{n}}$ & Partial matrix  of $\textbf{H}^{(n)}$  & Equ. (\ref{SGD_core_tensor});\\
$\bm{X}_{:,(\Psi^{(n)}_{M})_{i_{n}}}^{(n)}$ $\in$ $\mathbb{R}^{I_{n} \times |(\Psi^{(n)}_{M})_{i_{n}}|}$ & Partial matrix  of $\bm{X}^{(n)}$  & Equ. (\ref{SGD_factor_matrix});\\
$\textbf{E}^{(n)}_{:, (\Psi^{(n)}_{M})_{i_{n}}}$ $\in$ $\mathbb{R}^{J_{n} \times |(\Psi^{(n)}_{M})_{i_{n}}|}$ & Partial matrix of $\textbf{E}^{(n)}$    & Equ. (\ref{SGD_factor_matrix}).\\
 \hline
 \hline
\end{tabular}
\end{table}

\subsection{Optimization Process for the Core Tensor} \label{Section31}
Due to the non-trivial coupling of the core tensor $\bm{\mathcal{G}}$, the efficient and effective way to infer it is through an approximation.
A tensor $\bm{\mathcal{G}}$ can be approximated by $R_{core}$ $\leq$ ${J_{n}, n\in\{N\}}$
Kruskal product of low-rank matrices $\{\textbf{B}^{(n)}\in \mathbb{R}^{J_{n}\times R_{core}}|n\in \{N\}\}$
to form $\widehat{\bm{\mathcal{G}}}$:
 \begin{equation}\label{CP-Decomposition}
  \begin{aligned}
\widehat{\bm{\mathcal{G}}}=\sum_{r_{core}=1}^{R_{core}} b^{(1)}_{:,r_{core}}\circ\cdots\circ b^{(n)}_{:,r_{core}}\circ\cdots\circ b^{(N)}_{:,r_{core}}.
   \end{aligned}
\end{equation}

As the direct approximation for the core tensor $\bm{\mathcal{G}}$ may result in instability, we propose to apply the coupling process to approximate STD and tensor approximation.
Specifically, we use Kruskal product of low-rank matrices
$\{\textbf{B}^{(n)}\in \mathbb{R}^{J_{n}\times R_{core}}|n\in\{N\}\}$ as follows:
 \begin{equation}\label{approximated_core_tensor_optimization}
  \begin{aligned}
  \mathop{\arg\min}_{\widehat{\bm{\mathcal{G}}}}
  f&\bigg(\widehat{g}^{(n)}\bigg|x^{(n)}, \{\textbf{A}^{(n)}\}\bigg)\\
  =&\bigg\|x^{(n)}-\textbf{H}^{(n)}\widehat{g}^{(n)}\bigg\|_{2}^{2}+
  \lambda_{\widehat{g}^{(n)}}\bigg\|\widehat{g}^{(n)}\bigg\|_{2}^{2},
   \end{aligned}
\end{equation}
where
$\widehat{g}^{(n)}$ $=$ $Vec(\textbf{B}^{(n)}\textbf{Q}^{(n)^{T}})$ and
$\textbf{Q}^{(n)}$  $=$ $\textbf{B}^{(N)}\odot\cdots \odot\textbf{B}^{(n+1)}\odot\textbf{B}^{(n-1)}\odot\cdots \odot\textbf{B}^{(1)}$.
\newtheorem{theorem}{Theorem}
The tanglement problem resulted from $R_{core}$
Kruskal product of low-rank matrices $\{\textbf{B}^{(n)}\in \mathbb{R}^{J_{n}\times R_{core}}|n\in \{N\}\}$
leads to a non-convex optimization problem for the optimization objective (\ref{approximated_core_tensor_optimization}).
The alternative optimization strategy \cite{ex241} is adopted to update the parameters $\{\textbf{B}^{(n)}\in \mathbb{R}^{J_{n}\times R_{core}}|n\in\{N\}\}$ and
then the non-convex optimization objective (\ref{approximated_core_tensor_optimization}) is turned into the objective as:
 \begin{equation}\label{eventual0_approximated_core_tensor_optimization}
  \begin{aligned}
  \mathop{\arg\min}_{\textbf{B}^{(n)}, n\in\{N\}}
  &f\bigg(\textbf{B}^{(n)}\bigg|x^{(n)}, \{\textbf{A}^{(n)}\}, \{\textbf{B}^{(n)}\}\bigg)\\
  =&\bigg\|x^{(n)}-\textbf{H}^{(n)}Vec(\textbf{B}^{(n)}\textbf{Q}^{(n)^{T}})\bigg\|_{2}^{2}+\lambda_{\textbf{B}}\bigg\|\textbf{B}\bigg\|_{2}^{2},
   \end{aligned}
\end{equation}
where $Vec(\textbf{B}^{(n)}\textbf{Q}^{(n)^{T}})$ $=$ $Vec(\sum\limits_{r=1}^{R_{core}}b^{(n)}_{:,r}q_{:,r}^{(n)^{T}})$.

The optimization objective (\ref{eventual0_approximated_core_tensor_optimization}) is not explicit for the variable $\textbf{B}^{(n)}, n\in\{N\}$ and
it is hard to find the gradient of the variables $\textbf{B}^{(n)}, n\in\{N\}$ under the current formulation.
Thus, we borrow the intermediate variables with a unity matrix as
$\textbf{O}^{(n)}_{r}$ $\in$ $\mathbb{R}^{\prod_{k=1}^{N}J_{k}\times J_{n}}$, $r$ $\in$ $\{R_{core}\}$ and
the unity matrix
$\textbf{U}^{(n)}$ $\in$ $\mathbb{R}^{J_{n}\times J_{n}}$, $n\in\{N\}$ which can present the variables $\textbf{B}^{(n)}, n\in\{N\}$ in a gradient-explicit formation.
The matrix $\textbf{O}^{(n)}_{r}$ is defined as
 \begin{equation}\label{Q_r}
  \begin{aligned}
\textbf{O}^{(n)}_{r}=
\bigg[
&q_{1,r}\textbf{U}^{(n)},
\cdots,
q_{m,r}\textbf{U}^{(n)},\cdots, \\
&q_{\prod\limits_{k=1,k\neq n}^{N}J_{k},r}\textbf{U}^{(n)}
\bigg]^{T}.
   \end{aligned}
\end{equation}
The key transformation of the gradient-explicit formation for the variables $\textbf{B}^{(n)}, n\in\{N\}$ is as
$Vec(\sum\limits_{r=1}^{R_{core}}b^{(n)}_{:,r}q_{:,r}^{(n)^{T}}) =
\sum\limits_{r=1}^{R_{core}} \textbf{O}^{(n)}_{r}b^{(n)}_{:,r}$.
Then the optimization objective (\ref{eventual0_approximated_core_tensor_optimization}) is reformulated into:
 \begin{equation}\label{eventual1_approximated_core_tensor_optimization}
  \begin{aligned}
  \mathop{\arg\min}_{\textbf{B}^{(n)}, n\in\{N\}}
  &f\bigg(\textbf{B}^{(n)}\bigg|x^{(n)}, \{\textbf{A}^{(n)}\}, \{\textbf{B}^{(n)}\}\bigg)\\
  &=\bigg\|x^{(n)}-\textbf{H}^{(n)}\sum\limits_{r=1}^{R_{core}} \textbf{O}^{(n)}_{r}b^{(n)}_{:,r} \bigg\|_{2}^{2}+\lambda_{\textbf{B}}\bigg\|\textbf{B}\bigg\|_{2}^{2}.
   \end{aligned}
\end{equation}

The cyclic block optimization strategy \cite{ex218} is adopted to update the variables
$\big\{b^{(n)}_{:,r_{core}}|r_{core}\in \{R_{core}\}\big\}$ within a low-rank matrix $\textbf{B}^{(n)}, n\in \{N\}$ and eventually the optimization objective is reformulated into:
 \begin{equation}\label{eventual_approximated_core_tensor_optimization}
  \begin{aligned}
  \mathop{\arg\min}_{b^{(n)}_{:,r_{core}}, n\in\{N\}}
  &f\bigg(b^{(n)}_{:,r_{core}}\bigg|x^{(n)}, \{\textbf{A}^{(n)}\}, \{\textbf{B}^{(n)}\}\bigg)\\
  =&\bigg\|x^{(n)}_{r_{core}}-\textbf{H}^{(n)}_{r_{core}}b^{(n)}_{:,r_{core}} \bigg\|_{2}^{2}+\lambda_{\textbf{B}}\bigg\|b^{(n)}_{:,r_{core}}\bigg\|_{2}^{2},
\end{aligned}
\end{equation}
where
$x^{(n)}_{r_{core}}$ $=$ $x^{(n)}-\textbf{H}^{(n)}\sum\limits_{r=1,r\neq r_{core}}^{R_{core}} \textbf{O}^{(n)}_{r}b^{(n)}_{:,r}$ and
$\textbf{H}^{(n)}_{r_{core}}$ $=$ $\textbf{H}^{(n)}\textbf{O}^{(n)}_{r_{core}}\in \mathbb{R}^{\prod\limits_{n=1}^{N}I_{n}\times J_{n}}$.

\begin{theorem}
From the function form of (\ref{eventual_approximated_core_tensor_optimization}), the optimization objective for $b^{(n)}_{:,r_{core}}$ is a $u$-convex and $L$-smooth function.
\end{theorem}
\begin{proof}
The proof is divided into the following two parts:
1) The non-convex optimization problem is transformed into fixing factor matrices $\textbf{A}^{(n)}, n\in \{N\}$ and
updating core tensor part, and this transformation
can keep promise that convex optimization for each part and appropriate accuracy \cite{ex234, ex237, ex236, ex241}.
2) The distance function $\big\|x^{(n)}_{r_{core}}-\textbf{H}^{(n)}_{r_{core}}b^{(n)}_{:,r_{core}} \big\|_{2}^{2}$ is an Euclidean distance with
$L_{2}$ norm regularization $\lambda_{\textbf{B}}\big\|b^{(n)}_{:,r_{core}}\big\|_{2}^{2}$ \cite{ex255}.
Thus, the optimization objective (\ref{eventual_approximated_core_tensor_optimization}) is a $u$-convex and $L$-smooth function obviously \cite{ex241}, \cite{ex242}, \cite{ex243}.
The proof details are omitted which can be found on \cite{ex225}, \cite{ex244}.
\end{proof}

\subsection{Optimization Process for Factor Matrices} \label{Section32}
After the optimization step is conducted for $\{\textbf{B}^{(n)}\in \mathbb{R}^{J_{n}\times R_{core}}|n\in\{N\}\}$,
the Kruskal product for approximated core tensor $\hat{g}^{(n)}$ is constructed by Equ. (\ref{CP-Decomposition}).
Thus, we should consider the optimization problem for factor matrices $\textbf{A}^{(n)}, \{n|n\in \{N\}\}$ as:
 \begin{equation}\label{optimization_for_factor}
  \begin{aligned}
  \mathop{\arg\min}_{\textbf{A}^{(n)},n \in\{N\}}
  &f\bigg(\textbf{A}^{(n)}\bigg|\bm{X}^{(n)}, \big\{\textbf{A}^{(n)}\big\}, \widehat{\textbf{G}}^{(n)}\bigg)\\
  &=\bigg\|\bm{X}^{(n)}- \widehat{\bm{X}}^{(n)}\bigg\|_{2}^{2}+
  \lambda_{\textbf{A}}\bigg\|\textbf{A}^{(n)}\bigg\|_{2}^{2},
   \end{aligned}
\end{equation}
where
$\widehat{\bm{X}}^{(n)}$ $=$ $\textbf{A}^{(n)}\widehat{\textbf{G}}^{(n)}\textbf{S}^{(n)^{T}}$ and
$\textbf{S}^{(n)}$       $=$ $\textbf{A}^{(N)}\otimes\cdots \otimes\textbf{A}^{(n+1)}\otimes\textbf{A}^{(n-1)}\otimes\cdots \otimes\textbf{A}^{(1)}$.

The optimization process for the whole factor matrix set $\{\textbf{A}^{(1)},\cdots,\textbf{A}^{(N)}\}$ is non-convex.
The alternative optimization strategy \cite{ex241} is adopted to transform the non-convex optimization problem into $N$ convex optimization problems under
a fine-tuned initialization as:
\begin{equation}\label{Alternative_Optimization_for_Factor_Matrices}
\setlength{\abovedisplayskip}{1pt}
\setlength{\belowdisplayskip}{1pt}
\begin{aligned}
\left\{
\begin{aligned}
 \mathop{\arg\min}_{\textbf{A}^{(1)}} f\bigg(\textbf{A}^{(1)}\bigg)=
 &\bigg\|\bm{X}^{(1)}- \textbf{A}^{(1)}\textbf{E}^{(1)}\bigg\|_{2}^{2}+\lambda_{\textbf{A}}\bigg\|\textbf{A}^{(1)}\bigg\|_{2}^{2};\\
  \vdots\\
 \mathop{\arg\min}_{\textbf{A}^{(n)}} f\bigg(\textbf{A}^{(n)}\bigg)=
 &\bigg\|\bm{X}^{(n)}- \textbf{A}^{(n)}\textbf{E}^{(n)}\bigg\|_{2}^{2}+\lambda_{\textbf{A}}\bigg\|\textbf{A}^{(n)}\bigg\|_{2}^{2};\\
   \vdots\\
 \mathop{\arg\min}_{\textbf{A}^{(N)}} f\bigg(\textbf{A}^{(N)}\bigg)=
 &\bigg\|\bm{X}^{(N)}- \textbf{A}^{(N)}\textbf{E}^{(N)}\bigg\|_{2}^{2}+\lambda_{\textbf{A}}\bigg\|\textbf{A}^{(N)}\bigg\|_{2}^{2},
\end{aligned}
\right.
   \end{aligned}
\end{equation}
where $\textbf{E}^{(n)}$ $=$ $\widehat{\textbf{G}}^{(n)}\textbf{S}^{(n)^{T}}$ $\in$ $\mathbb{R}^{J_{n} \times \prod\limits_{k=1,k\neq n}^{N}I_{k}}$.
The optimization objectives of the $n$th variables are presented as an independent form as:
\begin{equation}\label{N_Optimization_for_Factor_Matrices}
\setlength{\abovedisplayskip}{1pt}
\setlength{\belowdisplayskip}{1pt}
\begin{aligned}
\left\{
\begin{aligned}
\mathop{\arg\min}_{a_{1,:}^{(n)}} f(a_{1,:}^{(n)})=
&\bigg\|\textbf{X}_{1,:}^{(n)}- a_{1,:}^{(n)}\textbf{E}^{(n)}\bigg\|_{2}^{2}+\lambda_{\textbf{A}}\bigg\|a_{1,:}^{(n)}\bigg\|_{2}^{2};\\
 \vdots\\
\mathop{\arg\min}_{a_{i_{n},:}^{(n)}} f(a_{i_{n},:}^{(n)})=
&\bigg\|\textbf{X}_{i_{n},:}^{(n)}- a_{i_{n},:}^{(n)}\textbf{E}^{(n)}\bigg\|_{2}^{2}+\lambda_{\textbf{A}}\bigg\|a_{i_{n},:}^{(n)}\bigg\|_{2}^{2};\\
 \vdots\\
\mathop{\arg\min}_{a_{I_{n},:}^{(n)}} f(a_{I_{n},:}^{(n)})=
&\bigg\|\textbf{X}_{I_{n},:}^{(n)}- a_{I_{n},:}^{(n)}\textbf{E}^{(n)}\bigg\|_{2}^{2}+\lambda_{\textbf{A}}\bigg\|a_{I_{n},:}^{(n)}\bigg\|_{2}^{2}.
\end{aligned}
\right.
   \end{aligned}
\end{equation}
\begin{theorem}From the function form of (\ref{N_Optimization_for_Factor_Matrices}), the optimization objective for $\textbf{A}^{(n)}$ is a $u$-convex and $L$-smooth function.
\end{theorem}
\begin{proof}
By adopting the alternative strategy \cite{ex234, ex237, ex236, ex241},
we fix $\widehat{\bm{\mathcal{G}}}$ and $\textbf{A}^{(k)}, k\neq n, k\in\{N\}$.
Then, we update $\textbf{A}^{(n)}$.
The distance function $\bigg\|\textbf{X}_{i_{n},:}^{(n)}- a_{i_{n},:}^{(n)}\textbf{E}^{(n)}\bigg\|_{2}^{2}$ is an Euclidean distance with
$L_{2}$ norm regularization $\lambda_{\textbf{A}}\bigg\|a_{i_{n},:}^{(n)}\bigg\|_{2}^{2}$, $i_{n}\in {I_{n}}$
\cite{ex255}.
Thus, the optimization objective (\ref{eventual_approximated_core_tensor_optimization}) is a $u$-convex and $L$-smooth function obviously \cite{ex241}, \cite{ex242}, \cite{ex243}.
Due to the limited space, the proof details are omitted which can be found on \cite{ex225}, \cite{ex244}.
\end{proof}

\subsection{Stochastic Optimization}\label{Section33}
The previous Sections~\ref{Section31} and~\ref{Section32} presented the transformation of the optimization problem.
In this section, the solvent is introduced.
The average SGD method is a basic part of state of the art stochastic optimization models.
Thus, in this section, we present the average SGD for our optimization objectives.
The optimization objectives for the core tensor are presented in Equs.
(\ref{eventual1_approximated_core_tensor_optimization}) and (\ref{eventual_approximated_core_tensor_optimization}).
The optimization objectives for the factor matrix are presented in Equ. (\ref{N_Optimization_for_Factor_Matrices}) which
are in the form of a basic optimization model.
In the industrial and big-data communities,
the HOHDST is very common.
Thus, SGD is proposed to replace the original optimization strategy.

The solution for $b^{(n)}_{:,r_{core}}$ is presented as:
 \begin{equation}\label{Solution_core_tensor}
  \begin{aligned}
  \mathop{\arg\min}_{b^{(n)}_{:,r_{core}}, n\in\{N\}}
  &f\bigg(b^{(n)}_{:,r_{core}}\bigg|x^{(n)}, \{\textbf{A}^{(n)}\}, \{\textbf{B}^{(n)}\}\bigg)\\
 =&\sum\limits_{i\in \Omega^{(n)}_{V}} L_{i}\bigg(b^{(n)}_{:,r_{core}}\bigg|x^{(n)}\bigg)+ \lambda_{\textbf{B}}\bigg\|b^{(n)}_{:,r_{core}}\bigg\|_{2}^{2}.
   \end{aligned}
\end{equation}
If a set $\Psi$ including $M$ randomly entries is selected,
the approximated SGD solution for $b^{(n)}_{:,r_{core}}$ is presented as:
 \begin{equation}\label{SGD_core_tensor}
  \begin{aligned}
&\mathop{\arg\min}_{b^{(n)}_{:,r_{core}}, n\in\{N\}}
f_{\Psi^{(n)}_{V}}\bigg(b^{(n)}_{:,r_{core}}\bigg|x^{(n)}_{\Psi^{(n)}_{V}}, \{\textbf{A}^{(n)}\}, \{\textbf{B}^{(n)}\}\bigg)\\
 &=\sum\limits_{i\in \Psi^{(n)}_{V}} L_{i}\bigg(b^{(n)}_{:,r_{core}}\bigg|x^{(n)}\bigg)+ \lambda_{\textbf{B}}\bigg\|b^{(n)}_{:,r_{core}}\bigg\|_{2}^{2}\\
 &=\bigg\|x^{(n)}_{\Psi^{(n)}_{V}}-\textbf{H}_{\Psi^{(n)}_{V},:}^{(n)}\sum\limits_{r=1}^{R_{core}} \textbf{O}^{(n)}_{r}b^{(n)}_{:,r} \bigg\|_{2}^{2}
  +\lambda_{\textbf{B}}\bigg\|b^{(n)}_{:,r_{core}}\bigg\|_{2}^{2}\\
 &=\bigg\|(x^{(n)}_{r_{core}})_{\Psi^{(n)}_{V}}-\textbf{H}_{\Psi^{(n)}_{V},:}^{(n)}\textbf{O}^{(n)}_{r_{core}}b^{(n)}_{:,r_{core}} \bigg\|_{2}^{2}+\lambda_{\textbf{B}}\bigg\|b^{(n)}_{:,r_{core}}\bigg\|_{2}^{2}.
   \end{aligned}
\end{equation}
The SGD for the approximated function $f_{\Psi^{(n)}_{V}}\bigg(b^{(n)}_{:,r_{core}}\bigg|x^{(n)}_{\Psi^{(n)}_{V}}, \{\textbf{A}^{(n)}\}, \{\textbf{B}^{(n)}\}\bigg)$ is deduced as:
\begin{equation}\label{Gradient_core_tensor}
  \begin{aligned}
  &\frac{\partial f_{\Psi^{(n)}_{V}}\bigg(b^{(n)}_{:,r_{core}}\bigg|x^{(n)}_{\Psi^{(n)}_{V}}, \{\textbf{A}^{(n)}\}, \{\textbf{B}^{(n)}\}\bigg)}{\partial b^{(n)}_{:,r_{core}}}\\
  =&\frac{1}{M}\bigg(
  -\textbf{O}^{(n)T}_{r_{core}}\textbf{H}_{\Psi^{(n)}_{V},:}^{(n)T}(x^{(n)}_{r_{core}})_{\Psi^{(n)}_{V}}\\
   &+\textbf{O}^{(n)T}_{r_{core}}\textbf{H}_{\Psi^{(n)}_{V},:}^{(n)T}\textbf{H}_{\Psi^{(n)}_{V},:}^{(n)}\textbf{O}^{(n)}_{r_{core}}b^{(n)}_{:,r_{core}}\bigg)
   +\lambda_{\textbf{B}}b^{(n)}_{:,r_{core}}.
   \end{aligned}
\end{equation}
The solution for factor matrices $a_{i_{n},:}^{(n)}$, $i_{n}$ $\in$ $\{I_{N}\}$, $n$ $\in$ $\{N\}$ is presented as:
 \begin{equation}\label{Solution_factor_matrix}
  \begin{aligned}
  \mathop{\arg\min}_{a_{i_{n},:}^{(n)}, n \in\{N\}}
  &f\bigg(a_{i_{n},:}^{(n)}\bigg|\bm{X}_{i_{n},:}^{(n)}, \big\{\textbf{A}^{(n)}\big\}, \widehat{\textbf{G}}^{(n)}\bigg)\\
  =&\sum\limits_{j\in (\Omega^{(n)}_{M})_{i_{n}}} L_{j}\bigg(a_{i_{n},:}^{(n)}\bigg|\bm{X}_{i_{n},j}^{(n)}\bigg)+\lambda_{\textbf{A}}\bigg\|a_{i_{n},:}^{(n)}\bigg\|_{2}^{2}.
   \end{aligned}
\end{equation}
If a set $\Psi$ including $M$ randomly chosen entries is selected,
the SGD solution for $a_{i_{n},:}^{(n)}$ can be expressed as:
 \begin{equation}\label{SGD_factor_matrix}
  \begin{aligned}
&\mathop{\arg\min}_{a_{i_{n},:}^{(n)}, n \in\{N\}}
 f_{\Psi^{(n)}_{M}}\bigg(a_{i_{n},:}^{(n)}\bigg|\bm{X}_{i_{n},(\Psi^{(n)}_{M})_{i_{n}}}^{(n)}, \big\{\textbf{A}^{(n)}\big\}, \widehat{\textbf{G}}^{(n)}\bigg)\\
  &=\sum\limits_{j\in (\Psi^{(n)}_{M})_{i_{n}}} L_{j}\bigg(a_{i_{n},:}^{(n)}\bigg|\bm{X}_{i_{n},j}^{(n)}\bigg)+\lambda_{\textbf{A}}\bigg\|a_{i_{n},:}^{(n)}\bigg\|_{2}^{2}\\
  &=\bigg\|\textbf{X}_{i_{n}, (\Psi^{(n)}_{M})_{i_{n}}}^{(n)}- a_{i_{n},:}^{(n)}\textbf{E}^{(n)}_{:, (\Psi^{(n)}_{M})_{i_{n}}}\bigg\|_{2}^{2}+\lambda_{\textbf{A}}\bigg\|a_{i_{n},:}^{(n)}\bigg\|_{2}^{2}.
   \end{aligned}
\end{equation}
The stochastic gradient for the approximated function $f_{\Psi^{(n)}_{V}}\bigg(a_{i_{n},:}^{(n)}\bigg|\bm{X}_{i_{n},(\Psi^{(n)}_{M})_{i_{n}}}^{(n)}, \big\{\textbf{A}^{(n)}\big\}, \textbf{G}^{(n)}\bigg)$ is deduced as:
\begin{equation}\label{Gradient_factor_matrix}
  \begin{aligned}
  &\frac{\partial f_{\Psi^{(n)}_{V}}\bigg(a_{i_{n},:}^{(n)}\bigg|\bm{X}_{i_{n},(\Psi^{(n)}_{M})_{i_{n}}}^{(n)}, \big\{\textbf{A}^{(n)}\big\}, \widehat{\textbf{G}}^{(n)}\bigg)}{\partial a_{i_{n},:}^{(n)}}
   \end{aligned}
\end{equation}
\begin{equation*}
  \begin{aligned}
  =&\frac{1}{M}\bigg(-\textbf{X}_{i_{n}, (\Psi^{(n)}_{M})_{i_{n}}}^{(n)}\textbf{E}^{(n)T}_{:, (\Psi^{(n)}_{M})_{i_{n}}}\\
   &+a_{i_{n},:}^{(n)}\textbf{E}^{(n)}_{:, (\Psi^{(n)}_{M})_{i_{n}}}\textbf{E}^{(n)T}_{:, (\Psi^{(n)}_{M})_{i_{n}}}\bigg)+\lambda_{\textbf{A}}a_{i_{n},:}^{(n)}.
   \end{aligned}
\end{equation*}
The computational details are presented in Algorithm~\ref{alg31}, which shows that
SGD$\_$Tucker is able to divide the high-dimension intermediate matrices
$\bigg\{\textbf{H}_{\Omega^{(n)}_{V},:}^{(n)}, \textbf{S}^{(n)}_{(\Omega^{(n)}_{M})_{i_{n}},:}, \textbf{E}^{(n)}_{:, (\Omega^{(n)}_{M})_{i_{n}}}\big| i_{n} \in\{I_{n}\}, n \in \{N\}\bigg\}$
into
small batches of intermediate matrices
$\bigg\{\textbf{H}_{\Psi^{(n)}_{V},:}^{(n)}, \textbf{S}^{(n)}_{(\Psi^{(n)}_{M})_{i_{n}},:}, \textbf{E}^{(n)}_{:, (\Psi^{(n)}_{M})_{i_{n}}}\big| i_{n} \in\{I_{n}\}, n \in \{N\}\bigg\}$.
We summarized all steps of SGD$\_$Tucker in Algorithm \ref{alg31}.

\begin{theorem}
From the function forms of (\ref{eventual_approximated_core_tensor_optimization}) and (\ref{N_Optimization_for_Factor_Matrices}),
the optimization objectives for the core tensor and factor matrics are both $u$-convex and $L$-smooth functions.
The stochastic update rules of (\ref{Gradient_core_tensor}) and (\ref{Gradient_factor_matrix}) can ensure the convergency of  alternative optimization objectives
(\ref{eventual_approximated_core_tensor_optimization}) and (\ref{N_Optimization_for_Factor_Matrices}), respectively.
\end{theorem}
\begin{proof}
The two function forms of (\ref{eventual_approximated_core_tensor_optimization}) and (\ref{N_Optimization_for_Factor_Matrices}) can be conclude as
$f(x)=\frac{1}{n}\mathop{\sum}_{i=1}^{n}f_{i}(\textbf{x})$,
where $f(\textbf{x})$ is a strongly-convex with constant $\mu$, and each $f_{i}(\textbf{x})$ is smooth and Lipschitz-continuous with constant $L$.
At $t$th iteration, for chosen $f_{i_{t}}$ randomly, and a learning rate sequence $\gamma_{t}>0$, the expectation $\mathbb{E}[\nabla f_{i}(\textbf{x}_{t})|\textbf{x}_{t}]$ of $\nabla f_{i_{t}}(\textbf{x}_{t})$ is equivalent to $\nabla f(\textbf{x}_{t})$ \cite{ex242, ex243, ex244, ex245}.
\end{proof}

\begin{algorithm}[!ht]
    \caption{Stochastic Optimization Strategies on a Training Epoch for SGD$\_$Tucker.}
    \label{alg31}
    \vspace{.1cm}
    $\textbf{Input}$: Sparse tensor $\mathcal{X}$, randomly selected set $\Psi$ with $M$ entries, Learning step $\gamma_{\textbf{A}}$, learning step $\gamma_{\textbf{B}}$; \\
    Initializing
    $\textbf{B}^{(n)}$, $n$ $\in$ $\{N\}$,
    core tensor $\widehat{\mathcal{G}}$,
    $\textbf{A}^{(n)}$, $n$ $\in$ $\{N\}$,
    $\textbf{H}_{\Psi^{(n)}_{V},:}^{(n)}$ $\in$ $\mathbb{R}^{M\times \prod_{n=1}^{N}J_{n}}$, $r$ $\in$ $\{R_{core}\}$, $n$ $\in$ $\{N\}$,
    $\textbf{O}^{(n)}_{r}$ $\in$ $\mathbb{R}^{\prod_{k=1}^{N}J_{k}\times J_{n}}$, $r$ $\in$ $\{R_{core}\}$, $n$ $\in$ $\{N\}$,
    $\widehat{x}^{(n)}_{\Psi^{(n)}_{V}} \in \mathbb{R}^{M}$. \\
    $\textbf{Output}$: $\textbf{B}^{(n)}$, $n$ $\in$ $\{N\}$, $\widehat{\mathcal{G}}$, $\textbf{A}^{(n)}$, $n$ $\in$ $\{N\}$.\\
    \begin{algorithmic}[1]
    \FOR{$n$ from $1$ to $N$}
        \STATE Obtain $\textbf{H}_{\Psi^{(n)}_{V},:}^{(n)}$;\\
        \FOR{$r_{core}$ from $1$ to $R_{core}$}
            \STATE Obtain $\textbf{O}^{(n)}_{r_{core}}$ by Equ. (\ref{Q_r}); \\
            \STATE Obtain $\textbf{W}^{(n)}_{r_{core}}$ $\leftarrow$ $\textbf{H}_{\Psi^{(n)}_{V},:}^{(n)}\textbf{O}^{(n)}_{r_{core}}$;\\
        \ENDFOR
        \FOR{$r_{core}$ from $1$ to $R_{core}$}
            \STATE $\widehat{x}^{(n)}_{\Psi^{(n)}_{V}}$ $\leftarrow$ $x^{(n)}_{\Psi^{(n)}_{V}}$;\\
            \FOR{$r$ from $1$ to $R_{core}$ ($r$ $\neq$ $r_{core}$)}
                \STATE $\widehat{x}^{(n)}_{\Psi^{(n)}_{V}}$ $\leftarrow$ $\widehat{x}^{(n)}_{\Psi^{(n)}_{V}} - \textbf{W}^{(n)}_{r}b^{(n)}_{:,r}$;\\
            \ENDFOR
            \STATE
            $\textbf{C}_{r_{core}}$ $\leftarrow$ $\textbf{W}^{(n)^{T}}_{r_{core}}\textbf{W}^{(n)}_{r_{core}}$;\\
            \STATE
            $\textbf{V}^{(n)}_{r_{core}}$ $\leftarrow$
            $-\textbf{W}^{(n)^{T}}_{r_{core}}\widehat{x}^{(n)}_{\Psi^{(n)}_{V}}+
            \textbf{C}_{r_{core}}b^{(n)}_{:,r_{core}}$;\\ 
            \STATE
            $b^{(n)}_{:,r_{core}}$ $\leftarrow$
            $SGD\bigg(
            M,
            \lambda_{\textbf{B}},
            \gamma_{\textbf{B}},
            b^{(n)}_{:,r_{core}},
            \textbf{V}^{(n)}_{r_{core}}
            \bigg)$;\\ 
        \ENDFOR
    \ENDFOR
    \STATE Obtain $\widehat{\bm{\mathcal{G}}}$ by Equ. (\ref{CP-Decomposition});\\
    \FOR{$n$ from $1$ to $N$}
        \FOR{$i_{n}$ from $1$ to $I_{n}$}
            \STATE Obtain $\textbf{S}^{(n)}_{(\Psi^{(n)}_{M})_{i_{n}},:}(cache_{S})$;\\
            \STATE $\textbf{E}^{(n)}_{:, (\Psi^{(n)}_{M})_{i_{n}}}$ $\leftarrow$ $\widehat{\textbf{G}}^{(n)}\textbf{S}^{(n)^{T}}_{(\Psi^{(n)}_{M})_{i_{n}},:}(cache_{E})$;\\
            \STATE
            $\textbf{C}^{(n)}$ $\leftarrow$ $\textbf{E}^{(n)}_{:, (\Psi^{(n)}_{M})_{i_{n}}}\textbf{E}^{(n)T}_{:, (\Psi^{(n)}_{M})_{i_{n}}}$;\\
            \STATE
            $\textbf{F}^{(n)}$ $\leftarrow$
            $\underbrace{-\textbf{X}_{i_{n}, (\Psi^{(n)}_{M})_{i_{n}}}^{(n)}\textbf{E}^{(n)T}_{:, (\Psi^{(n)}_{M})_{i_{n}}}}_{cache_{Fact1}}+\overbrace{\underbrace{a_{i_{n},:}^{(n)}\textbf{C}^{(n)}}_{cache_{Fact2}}}^{cache_{Factp}, cache_{Factvec}}$; \\ 
            \STATE $a_{i_{n},:}^{(n)}$ $\leftarrow$
            $SGD\bigg( |(\Psi^{(n)}_{M})_{i_{n}}|,
            \lambda_{\textbf{A}},
            \gamma_{\textbf{A}},
             a_{i_{n},:}^{(n)},
             \textbf{F}^{(n)}
             \bigg)$. 
        \ENDFOR
    \ENDFOR
    \STATE $\textbf{Return}$: $\textbf{B}^{(n)}$, $n$ $\in$ $\{N\}$, $\widehat{\mathcal{G}}$, $\textbf{A}^{(n)}$, $n$ $\in$ $\{N\}$.\\
    \end{algorithmic}
\end{algorithm}

\subsection{Parallel and Distributed Strategy} \label{Section35}
\textcolor[rgb]{0.00,0.00,1.00}{We first explain the naive parallel strategy (Section \ref{Section351}), which relies on the coordinate continuity and
the specific style of matricization unfolding of the sparse tensor to keep the
entire continuous accessing in the process of updating $\textbf{A}^{(n)}$, $n$ $\in$ $\{N\}$.
Then, we present the improved parallel strategy (Section \ref{Section352}),
which can save the $N$ compressive styles of the sparse tensor $\mathcal{X}$ to just a compressive format.
At last, the analysis of the communication cost is reported on Section \ref{Section353}.}

\subsubsection{Naive Parallel Strategy} \label{Section351}
To cater to the need of processing big data, the algorithm design shall leverage the increasing number of cores while minimizing the communication overehead among parallel threads.
If there are $L$ threads,
the randomly selected set $\Psi$ is divided into $L$ subsets $\big\{\Psi_{1},\cdots,\Psi_{l},\cdots,\Psi_{L}\big|l\in\{L\}\big\}$ and
the parallel computation analysis follows each step in Algorithm \ref{alg31} as:

(i) Updating the core tensor:

\textbf{Line} 2:
The computational tasks of the $L$ Intermediate matrices
$\bigg\{\textbf{H}_{\Psi_{l_{V},:}^{(n)}}^{(n)}$ $\in$ $\mathbb{R}^{|\Psi_{l}|\times \prod\limits_{n=1}^{N}J_{n}}\bigg|l \in \{L\}, n\in \{N\}\bigg\}$ be allocated to $L$ threads;

\textbf{Line} 4:
Intermediate matrix $\textbf{O}^{(n)}_{r}$ $\in$ $\mathbb{R}^{\prod\limits_{k=1}^{N}J_{k}\times J_{n}}$, $r$ $\in$ $\{R_{core}\}$, $n$ $\in$ $\{N\}$;
thus, the independent computational tasks of the $\prod\limits_{k=1, k\neq n}^{N}J_{k}$ diagonal sub-matrices can be allocated to the $L$ threads;

\textbf{Line} $5$:
The computational tasks of the $L$ intermediate matrices
$\bigg\{\textbf{H}_{\Psi_{l_{V},:}^{(n)}}^{(n)}\textbf{O}^{(n)}_{r}$ $\in$ $\mathbb{R}^{|\Psi_{l}|\times J_{n}}\bigg| l \in \{L\}, n\in \{N\}\bigg\}$ can be allocated to $L$ threads;

\textbf{Line} $8$:
The $L$ assignment tasks
$\bigg\{\widehat{x}^{(n)}_{\Psi^{(n)}_{l_{V}}} \bigg| l \in \{L\}\bigg\}$
can be allocated to the $L$ threads;

\begin{figure}[ht]
\centering
\includegraphics[width=3.5in,height=2.5in]{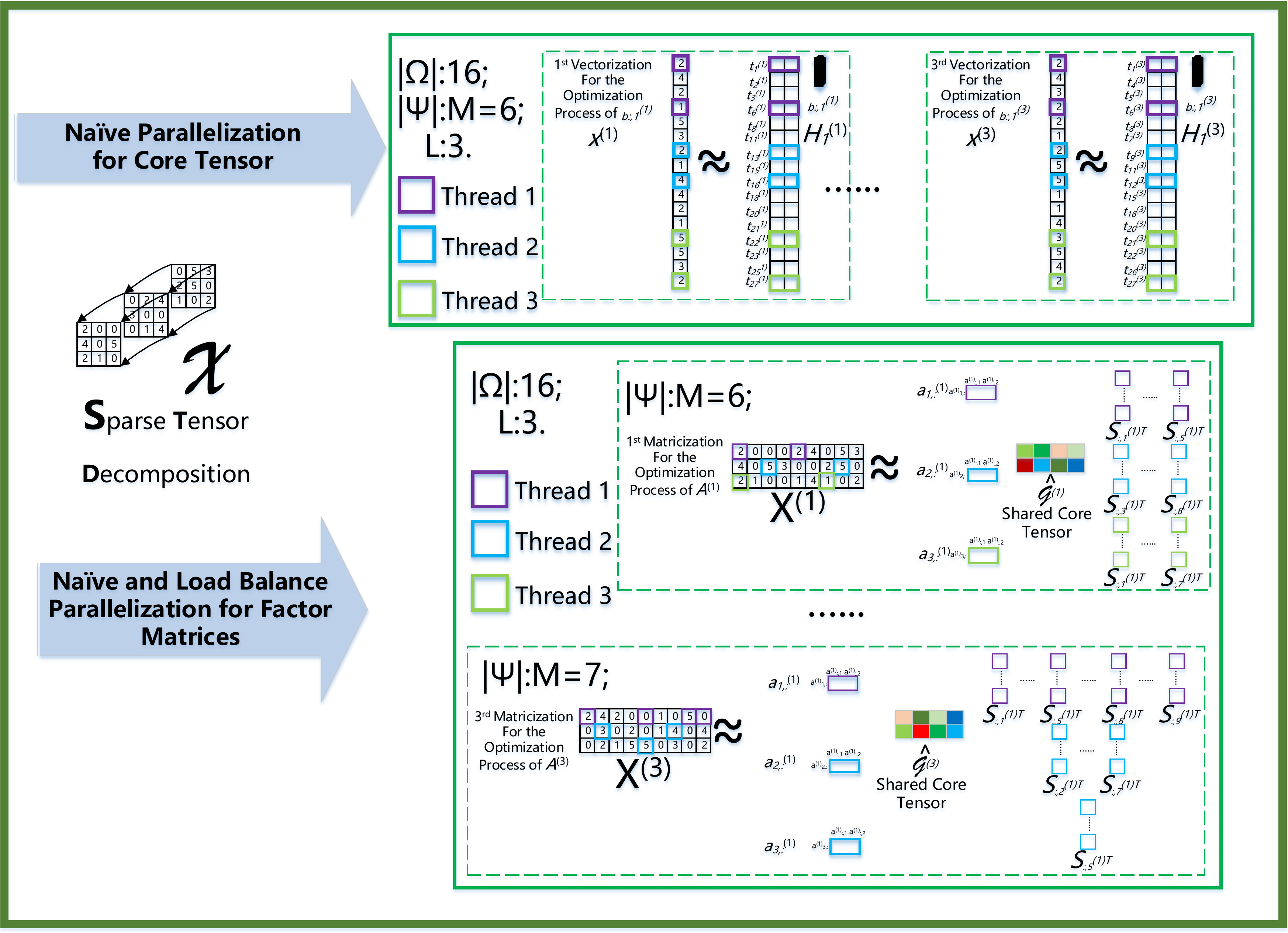}
\caption{Naive Parallel Strategy for SGD$\_$Tucker.}
\label{fig303}
\end{figure}

\textbf{Line} $10$:
The computational tasks of the $L$ intermediate matrices
$\bigg\{(\textbf{W}_{r}^{(n)})_{\Psi_{l_{V},:}^{(n)}}b^{(n)}_{:,r} \in \mathbb{R}^{|\Psi_{l}|} \bigg| r\in\{R_{core}\}, n\in \{N\}\bigg\}$ can be allocated to
$L$ threads;

\textbf{Line} $12$:
The computational tasks of the $L$ intermediate matrices
$\bigg\{
\textbf{C}_{r_{core}}^{l}=(\textbf{W}_{r}^{(n)})^{T}_{\Psi_{l_{V},:}^{(n)}}(\textbf{W}_{r}^{(n)})_{\Psi_{l_{V},:}^{(n)}}\in \mathbb{R}^{J_{n}\times J_{n}}
\bigg|
l \in\{L\},
r\in \{R_{core}\}, n\in \{N\}
\bigg\}$ can be allocated to $L$ threads and the $l$ thread does the intra-thread summation $\textbf{C}_{r_{core}}^{l}$ $=$ $(\textbf{W}_{r}^{(n)})^{T}_{\Psi_{l_{V},:}^{(n)}}(\textbf{W}_{r}^{(n)})_{\Psi_{l_{V},:}^{(n)}}$.
Then, the main thread sums $\textbf{C}_{r_{core}}= \sum\limits_{l=1}^{L}\textbf{C}_{r_{core}}^{l}$;

\textbf{Line} $13$:
The computational tasks of the $L$ intermediate matrices
$\bigg\{
(\textbf{W}_{r}^{(n)})^{T}_{\Psi_{l_{V},:}^{(n)}}\widehat{x}^{(n)}_{\Psi_{l_{V}}^{(n)}} \in \mathbb{R}^{J_{n}}$,
$\textbf{C}_{r}b^{(n)}_{:,r} \in \mathbb{R}^{J_{n}}
\bigg|
r\in \{R_{core}\}, n\in \{N\}
\bigg\}$ can be allocated to $L$ threads.

\textbf{Line} $14$: This step can be processed by the main thread.

(ii) Updating factor matrices:
$I_{n}$ loops including the \textbf{Lines} $20-24$, $n$ $\in$ $\{N\}$ are independent.
Thus, the $I_{n}$, $n$ $\in$ $\{N\}$ loops can be allocated to $L$ threads.

The parallel training process for $\textbf{B}^{(n)}$, $n$ $\in$ $\{N\}$ does not need the load balance.
The computational complexity of $\textbf{S}^{(n)}_{(\Psi^{(n)}_{M})_{i_{n}},:}$ for each thread is proportional to $|(\Psi^{(n)}_{M})_{i_{n}}|$, and
the $I_{n}$, $n$ $\in$ $\{N\}$ independent tasks are allocated to the $L$ threads.
Then, there is load imbalance problem for updating the factor matrices $\textbf{A}^{(n)}$, $n$ $\in$ $\{N\}$.
Load balance can fix the problem of non-even distribution for non-zero elements.

A toy example is shown in Fig. \ref{fig303} which comprises of:

(i) Each thread $l$ first selects the even number of non-zero elements:
The $3$ threads $\{1, 2, 3\}$ select
$\big\{$
$\{$ $\bm{x}^{(1)}_{1},   \bm{x}^{(1)}_{6}$ $\}$,
$\{$ $\bm{x}^{(1)}_{13},  \bm{x}^{(1)}_{16}$ $\}$,
$\{$ $\bm{x}^{(1)}_{22},  \bm{x}^{(1)}_{27}$ $\}$
$\big\}$, respectively.
\textbf{Step 1}:
The $3$ threads $\{1, 2, 3\}$ construct
$\big\{$
$\bm{H}^{(1)}_{1,:}$,
$\bm{H}^{(1)}_{13,:}$,
$\bm{H}^{(1)}_{22,:}$
$\big\}$, respectively.
Then, the $3$ threads $\{1, 2, 3\}$ compute
$\big\{$
$\textbf{W}^{(1)}_{1}=\bm{H}^{(1)}_{1,:}\textbf{O}^{(1)}$,
$\textbf{W}^{(1)}_{13}=\bm{H}^{(1)}_{13,:}\textbf{O}^{(1)}$,
$\textbf{W}^{(1)}_{22}=\bm{H}^{(1)}_{22,:}\textbf{O}^{(1)}$
$\big\}$, respectively.
\textbf{Step 2}:
The $3$ threads $\{1, 2, 3\}$ construct
$\big\{$
$\bm{H}^{(1)}_{6,:}$,
$\bm{H}^{(1)}_{16,:}$,
$\bm{H}^{(1)}_{27,:}$
$\big\}$, respectively.
Then, the $3$ threads $\{1, 2, 3\}$ compute
$\big\{$
$\textbf{W}^{(1)}_{6}=\bm{H}^{(1)}_{6,:}\textbf{O}^{(1)}$,
$\textbf{W}^{(1)}_{16}=\bm{H}^{(1)}_{16,:}\textbf{O}^{(1)}$,
$\textbf{W}^{(1)}_{27}=\bm{H}^{(1)}_{27,:}\textbf{O}^{(1)}$
$\big\}$, respectively.
Each thread does the summation within the thread and the $3$ threads do the entire summation by
the code $\#\ pragma\ omp\ parallel\ for\ reduction\ (+:sum)$
for multi-thread summation.
We observe that the computational process of $\textbf{W}^{(1)^{T}}_{k}\textbf{W}^{(1)}_{k}b^{(1)}_{1,:}$ , $k$ $\in$ $\{1, 6, 13, 16, 22, 27\}$
is divided into the vectors reduction operation $p=\textbf{W}^{(1)}_{k}b^{(1)}_{1,:}$ and $vec=\textbf{W}^{(1)^{T}}_{k}p$.
\textbf{Step 3}: The $b^{(1)}_{:,1}$ is updated.
The process of updating $b^{(3)}_{:,1}$ is similar the process of updating $b^{(1)}_{:,1}$.
The description is omitted.
We observe that each thread selects $2$ elements.
Thus, the load for the $3$ threads is balanced.

(ii) Each thread selects the independent rows and $L$ threads realize naive parallelization for $\textbf{A}^{(n)}$, $n$ $\in$ $\{N\}$ by the $n$th matricization $\textbf{X}^{(n)}$, $n$ $\in$ $\{N\}$:
As show in Fig. \ref{fig303}, the $3$ threads $\{1, 2, 3\}$ update
$\{a_{1,:}^{(1)}, a_{2,:}^{(1)}, a_{3,:}^{(1)}\}$ explicitly.
Thus, the description is omitted.
The $2$ threads $\{1, 2\}$ update
$\bigg\{a_{1,:}^{(1)}, \big\{a_{2,:}^{(1)}, a_{3,:}^{(1)}\big\}\bigg\}$, respectively and independently by
$\bigg\{$
$\big\{$
$\{\bm{X}^{(3)}_{1,1}$, $\bm{S}^{(3)}_{:,1}\}$, $\{\bm{X}^{(3)}_{1,5}$, $\bm{S}^{(3)}_{:,5}\}$, $\{\bm{X}^{(3)}_{1,8}$, $\bm{S}^{(3)}_{:,8}$, $\{\bm{X}^{(3)}_{1,9}$, $\bm{S}^{(3)}_{:,9}\}$
$\big\}$,
$\big\{$
$\{\bm{X}^{(3)}_{2,2}$, $\bm{S}^{(3)}_{:,2}\}$, $\{\bm{X}^{(3)}_{2,7}$, $\bm{S}^{(3)}_{:,7}\}$
$\big\}$,
$\big\{$
$\{\bm{X}^{(3)}_{3,5}$, $\bm{S}^{(3)}_{:,5}\}$ $\big\}$
$\bigg\}$, respectively, with the shared matrix $\widehat{\textbf{G}}^{(3)}$.
Thread $1$ selects $4$ elements for updating $a_{1,:}^{(3)}$.
Thread $2$ selects $2$ elements and $1$ element for updating $a_{2,:}^{(3)}$ and $a_{3,:}^{(3)}$, respectively,
which can dynamically  balance the load.
In this condition, the load for the $3$ threads is slightly more balanced.

As shown in Fig. \ref{fig303},
the naive parallel strategy for $\{ \textbf{A}^{(1)}, \cdots , \textbf{A}^{(N)} \}$ relies on
the matricization format $\{ \bm{X}^{(1)}, \cdots , \bm{X}^{(N)} \}$ of a sparse tensor $\mathcal{X}$, respectively,
which is used to avoid the read-after-write and write-after-read conflicts.
Meanwhile, the transformation process for the compressive format to another one is time consuming.
Thus, the space and computational overheads are not scalable.


\begin{figure}[ht]
\centering
\includegraphics[width=3.5in,height=2.0in]{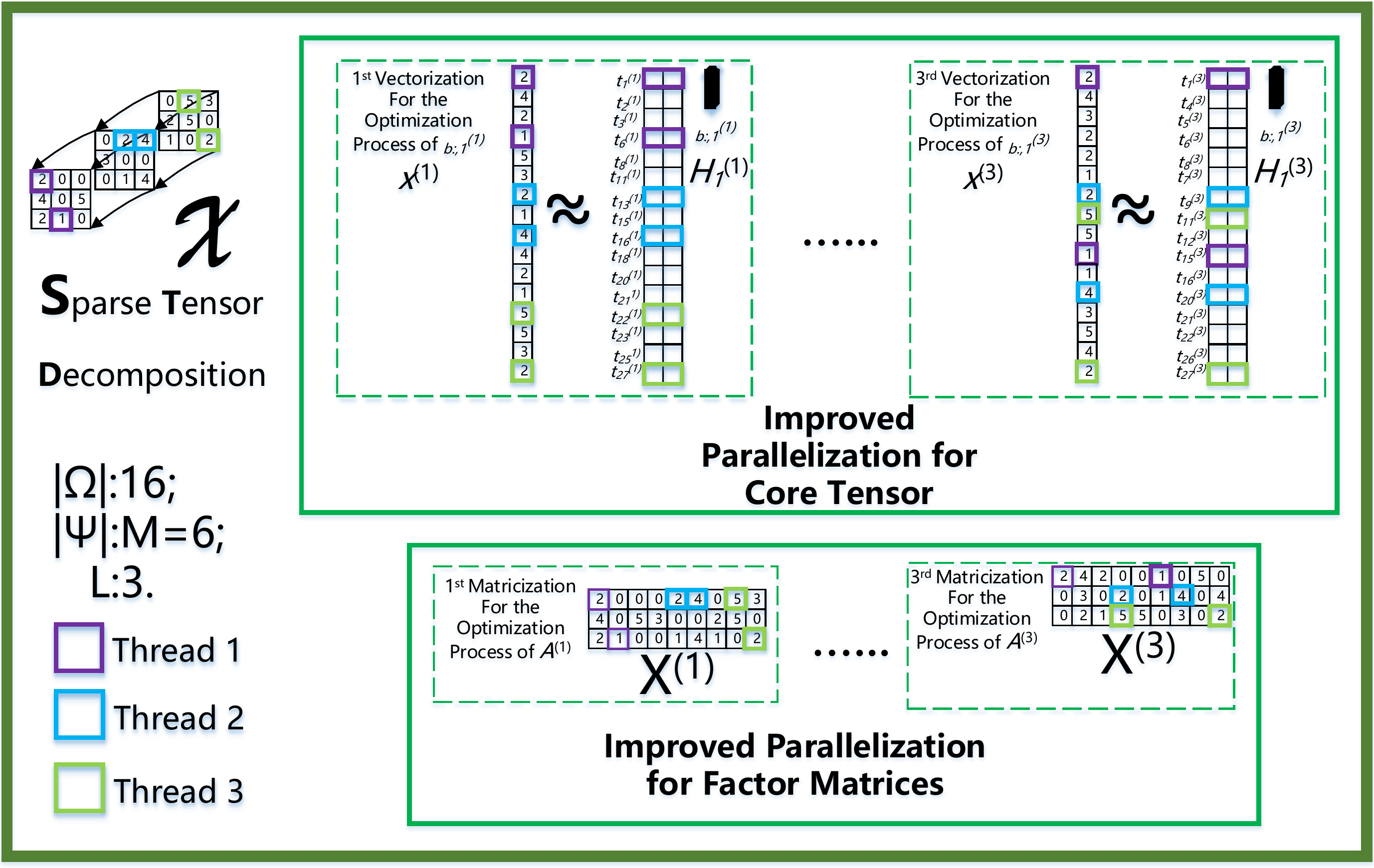}
\caption{Improved Parallel Strategy for SGD$\_$Tucker.}
\label{fig302}
\end{figure}

\subsubsection{Improved Parallel Strategy} \label{Section352}
The improved parallel strategy is developed to use only one  compressive format for the entire updating process and hence save the memory requirements. At the same time,  it can avoid the read-after-write or write-after-read conflicts.
In Fig. \ref{fig302}, a toy example of improved parallel strategy is illustrated.
As show in Fig. \ref{fig302},
the $3$ threads $\{ 1 , 2 , 3 \}$ select $6$ non-zeros elements and the
Coordinate Format (COO) is
$\bigg\{$
$\big\{$ $(1,1,1,2.0)$, $(3,2,1,1.0)$ $\big\}$,
$\big\{$ $(1,2,2,2.0)$, $(1,3,2,4.0)$ $\big\}$,
$\big\{$ $(1,2,3,5.0)$, $(3,3,3,2.0)$ $\big\}$
$\bigg\}$, respectively.
By the structure of COO,
the training process of $B^{(1)}$ and $B^{(3)}$
does not need a specific shuffling order of a sparse tensor.
Thus, the description of updating $b_{:, 1}^{(1)}$ and $b_{:, 1}^{(3)}$ is omitted.

As for  $\bm{X}^{(3)}$ to update $\textbf{A}^{(3)}$,
we neglect this condition because the selected elements of  $3$ threads lie in independent rows and it is trivial to parallelize.
In the style of $\bm{X}^{(1)}$ for updating $\textbf{A}^{(1)}$,
updating $a_{1, :}^{(1)}$ relies on
$\big\{ (1,1,1,2.0), \textbf{S}_{:, 1}^{(1)} \big\}$ (selected by thread $1$),
$\bigg\{ \big\{(1,2,2,2.0), \textbf{S}_{:, 5}^{(1)}\big\}, \big\{(1,3,2,4.0), \textbf{S}_{:, 6}^{(1)}\big\} \bigg\}$ (selected by thread $2$), and
$\big\{(1,2,3,5.0), \textbf{S}_{:, 8}^{(1)}\big\}$ (selected by thread $3$) with the shared matrix $\widehat{\textbf{G}}^{(1)}$.
It has following three steps.

(i) (Lines 20-21 in Algorithm \ref{alg31})
The $3$ threads $\{1, 2, 3\}$ compute
$\bigg\{$
$ \textbf{E}_{:, 1}^{(1)} = \widehat{\textbf{G}}^{(1)} \textbf{S}_{:, 1}^{(1)^{T}} $,
$\big\{
\textbf{E}_{:, 5}^{(1)} = \widehat{\textbf{G}}^{(1)} \textbf{S}_{:, 5}^{(1)^{T}},
\textbf{E}_{:, 6}^{(1)} = \widehat{\textbf{G}}^{(1)} \textbf{S}_{:, 6}^{(1)^{T}}
\}$,
$\textbf{E}_{:, 8}^{(1)} = \widehat{\textbf{G}}^{(1)} \textbf{S}_{:, 8}^{(1)^{T}}$
$\bigg\}$ independently, by the \textbf{private} \textbf{cache} \textbf{matrix} $\{cache_{S}, cache_{E}\}$ of each thread;

(ii) (Lines 22-23 in Algorithm \ref{alg31})
The computational process of $a_{i_{n},:}^{(n)}\textbf{E}_{:, k}^{(1)} \textbf{E}_{:, k}^{(1)^{T}}$, $k$ $\in$ $\{1, 5, 6, 8\}$
is divided into the vectors reduction operation $cache_{Factp}=a_{i_{n},:}^{(n)}\textbf{E}_{:, k}^{(1)}$ and $cache_{Factvec}=cache_{Factp}\textbf{E}_{:, k}^{(1)^{T}}$.
The $3$ threads $\{1, 2, 3\}$ compute
$\bigg\{$
$ a_{i_{n},:}^{(n)}\textbf{E}_{:, 1}^{(1)} \textbf{E}_{:, 1}^{(1)^{T}} $,
$\big\{
a_{i_{n},:}^{(n)}\textbf{E}_{:, 5}^{(1)} \textbf{E}_{:, 5}^{(1)^{T}} ,
a_{i_{n},:}^{(n)}\textbf{E}_{:, 6}^{(1)} \textbf{E}_{:, 6}^{(1)^{T}}
\}$,
$a_{i_{n},:}^{(n)}\textbf{E}_{:, 8}^{(1)} \textbf{E}_{:, 8}^{(1)^{T}}$
$\bigg\}$, respectively and independently, by the \textbf{private} \textbf{cache} $cache_{Factp}$ and $cache_{Factvec}$  of each thread.
Then, the $3$ threads $\{1, 2, 3\}$ can use the synchronization mechanisms, i.e., $atomic$, $cirtical$ or
$mutex$, of OpenMP to accomplish $\prod\limits_{k=1,5,6,8}a_{i_{n},:}^{(n)}\textbf{E}_{:, k}^{(1)}\textbf{E}_{:, k}^{(1)^{T}}$.
Then, the results are returned to \textbf{global} \textbf{shared} \textbf{cache} $cache_{Fact2}$;

(iii) (Line 23 in Algorithm \ref{alg31}).
The $3$ threads $\{1, 2, 3\}$ compute
$\bigg\{$
$ x_{1,1,1}\textbf{E}_{:, 1}^{(1)} $,
$\big\{
x_{1,2,2}\textbf{E}_{:, 5}^{(1)},
x_{1,3,2}\textbf{E}_{:, 6}^{(1)}
\}$,
$x_{1,2,3}\textbf{E}_{:, 8}^{(1)}$
$\bigg\}$, respectively and independently.
Then,
the $3$ threads $\{1, 2, 3\}$ can use the synchronization mechanisms, i.e., $atomic$, $cirtical$ or
$mutex$, to accomplish $\textbf{F}^{(1)}_{1}$.
Then, the results are returned to the \textbf{global} \textbf{shared} \textbf{cache} $cache_{Fact1}$.
Eventually, the $I_{1}$ tasks
$SGD\bigg( 6, \lambda_{\textbf{A}}, \gamma_{\textbf{A}}, a^{(1)}_{1,:}, \textbf{F}^{(1)}_{1} \bigg)$ be allocated to the
the $3$ threads $\{1, 2, 3\}$ in a parallel and \textbf{load} \textbf{balance} way.
Due to the same condition of updating $a_{3, :}^{(1)}$ and limited space,
the description of updating $a_{3, :}^{(1)}$ is omitted.

By the \textbf{global} \textbf{shared} \textbf{caches}
and \textbf{private} \textbf{caches},
$SGD$\_$Tucker$ can handle the parallelization on OpenMP by just a compressive format and
the space overhead is much less than the compressive data structure of a sparse tensor $\mathcal{X}$;
meanwhile, this strategy does not increase extra computational overhead.

\subsubsection{Communication in Distributed Platform} \label{Section353}
In distributed platform,
the communication overhead for a core tensor is
$O\bigg(\prod\limits_{n=1}^{N}J_{n}\bigg)$,
which is non-scalable in HOHDST scenarios.
SGD$\_$Tucker can prune the number of the parameters for constructing an updated core tensor
from $O\bigg(\prod\limits_{n=1}^{N}J_{n}\bigg)$ to $O\bigg(\sum\limits_{n=1}^{N}J_{n}R_{core}\bigg)$ where $R_{core}$ $\ll$ $J_{n}, n\in \{N\}$
while maintaining the same overall accuracy and lower computational complexity.
Hence, nodes only need to communicate the Kruckal product matrices
$\bigg\{\textbf{B}_{r}^{(n)} \in \mathbb{R}^{J_{n}\times R_{core}}\bigg| r\in \{R_{core}\}, n\in
\{N\} \bigg\}$ rather than the entire core tensor $\mathcal{G}$ $\in$ $\mathbb{R}^{J_{1}\times J_{2}\times\cdots \times J_{N}}_{+}$.
Hence,
SGD$\_$Tucker features that
1) a stochastic optimization strategy is adopted to core tensor and factor matrices which can keep the low computation overhead and storage space while not compromising the accuracy;
2) each minor node can select sub-samples from allocated sub-tensor and then compute the partial gradient and the major node gathers the partial gradient from all minor nodes and then obtains and returns the full gradient;
3) the communication overhead for information exchange of the core tensor $\bm{\mathcal{G}}$ is $O\bigg(\sum\limits_{n=1}^{N}J_{n}R_{core}\bigg)$.

\subsection{Complexity Analysis}\label{Section34}
The space and complexity analyses follow the steps which are presented in Algorithm \ref{alg31} as:
\begin{theorem}
The space overhead for updating $\textbf{B}^{(n)}$, $n$ $\in$ $\{N\}$ is
$O\bigg(
\big(MN+R_{core}NJ_{n}\big)\prod\limits_{k=1}^{N}J_{k}+
R_{core}MJ_{n}+
MN+
J_{n}
\bigg)$.
\end{theorem}

\begin{theorem}
The computational complexity for updating $\textbf{B}^{(n)}$, $n$ $\in$ $\{N\}$ is
$O\bigg(
(MN+R_{core}NJ_{n})\prod\limits_{k=1}^{N}J_{k}+
R_{core}MJ_{n}+
R_{core}(R_{core}-1)MN+
R_{core}NJ_{n}
\bigg)$.
\end{theorem}

\begin{proof}
In the process updating $\textbf{B}^{(n)}$, $n$ $\in$ $\{N\}$,
the space overhead and computational complexity of intermediate matrices
$\bigg\{
\textbf{H}_{\Psi^{(n)}_{V},:}^{(n)},
\textbf{O}^{(n)}_{r_{core}},
\textbf{W}^{(n)}_{r_{core}},
\widehat{x}^{(n)}_{\Psi^{(n)}_{V}},
\textbf{V}^{(n)}_{r_{core}}
\bigg| n\in\{N\}, r_{core}\in \{R_{core}\} \bigg\}$ are
$\bigg\{
M\prod\limits_{k=1}^{N}J_{k},
J_{n}\prod\limits_{k=1}^{N}J_{k},
MJ_{n},
M,
J_{n}
\bigg| n\in\{N\}, r_{core}\in \{R_{core}\} \bigg\}$ and
$\bigg\{
M\prod\limits_{k=1}^{N}J_{k},
\prod\limits_{k=1}^{N}J_{k},
MJ_{n}
\prod\limits_{k=1}^{N}J_{k}$,
$(R_{core}-1)MJ_{n},
MJ_{n}+MJ_{n}^{2}
\bigg| n\in\{N\}, r_{core}\in \{R_{core}\} \bigg\}$, respectively.
\end{proof}

\begin{theorem}
The space overhead for updating $\textbf{A}^{(n)}$, $n$ $\in$ $\{N\}$ is
$O\bigg(
(max(|(\Psi^{(n)}_{M})_{i_{n}}|)+1)\prod\limits_{k=1}^{N}J_{k}+
max(|(\Psi^{(n)}_{M})_{i_{n}}|)J_{n}+
J_{n}
\bigg)$.
\end{theorem}

\begin{theorem}
The computational complexity for updating $\textbf{A}^{(n)}$, $n$ $\in$ $\{N\}$ is
$O\bigg(
(N+MN+M)\prod\limits_{n=1}^{N}J_{n}+
\sum\limits_{n=1}^{N}I_{n}J_{n}
\bigg)$.
\end{theorem}

\begin{proof}
In the process updating $\textbf{A}^{(n)}$, $n$ $\in$ $\{N\}$,
the space overhead and computational complexity of intermediate matrices
$\bigg\{
\widehat{\textbf{G}}^{(n)},
\textbf{S}^{(n)}_{(\Psi^{(n)}_{M})_{i_{n}},:},
\textbf{E}^{(n)}_{:, (\Psi^{(n)}_{M})_{i_{n}}},
\textbf{F}^{(n)}
\bigg| i_{n} \in \{I_{n}\}, n\in\{N\}
\bigg\}$ are
$\bigg\{
\prod\limits_{k=1}^{N}J_{k},
max(|(\Psi^{(n)}_{M})_{i_{n}}|)\prod\limits_{k=1}^{N}J_{k},
max(|(\Psi^{(n)}_{M})_{i_{n}}|)J_{n},
J_{n}
\bigg| i_{n} \in \{I_{n}\}, n\in\{N\}
\bigg\}$ and
$\bigg\{
\prod\limits_{k=1}^{N}J_{k},
|(\Psi^{(n)}_{M})_{i_{n}}|\prod\limits_{k=1}^{N}J_{k},
|(\Psi^{(n)}_{M})_{i_{n}}|J_{n}$,
$J_{n}
\bigg| i_{n} \in \{I_{n}\}, n\in\{N\}
\bigg\}$, respectively.
\end{proof}

\begin{figure}[htp]
  \centering
    \subfigure[Movielen-10M]{
    \label{fig401 (a3)} 
    \includegraphics[width=1.4in,height=1.4in]{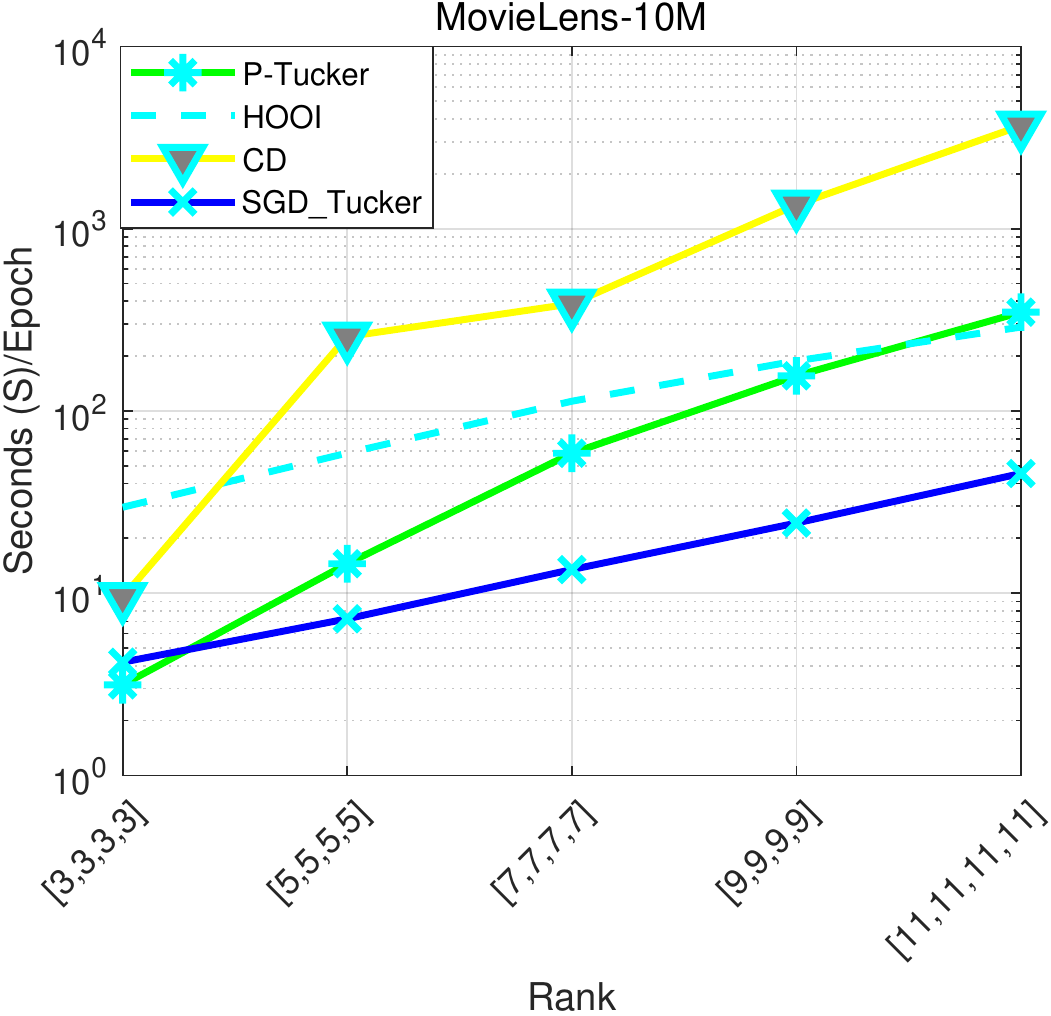}}
    ~
    \subfigure[Movielen-20M]{
    \label{fig401 (a4)} 
    \includegraphics[width=1.4in,height=1.4in]{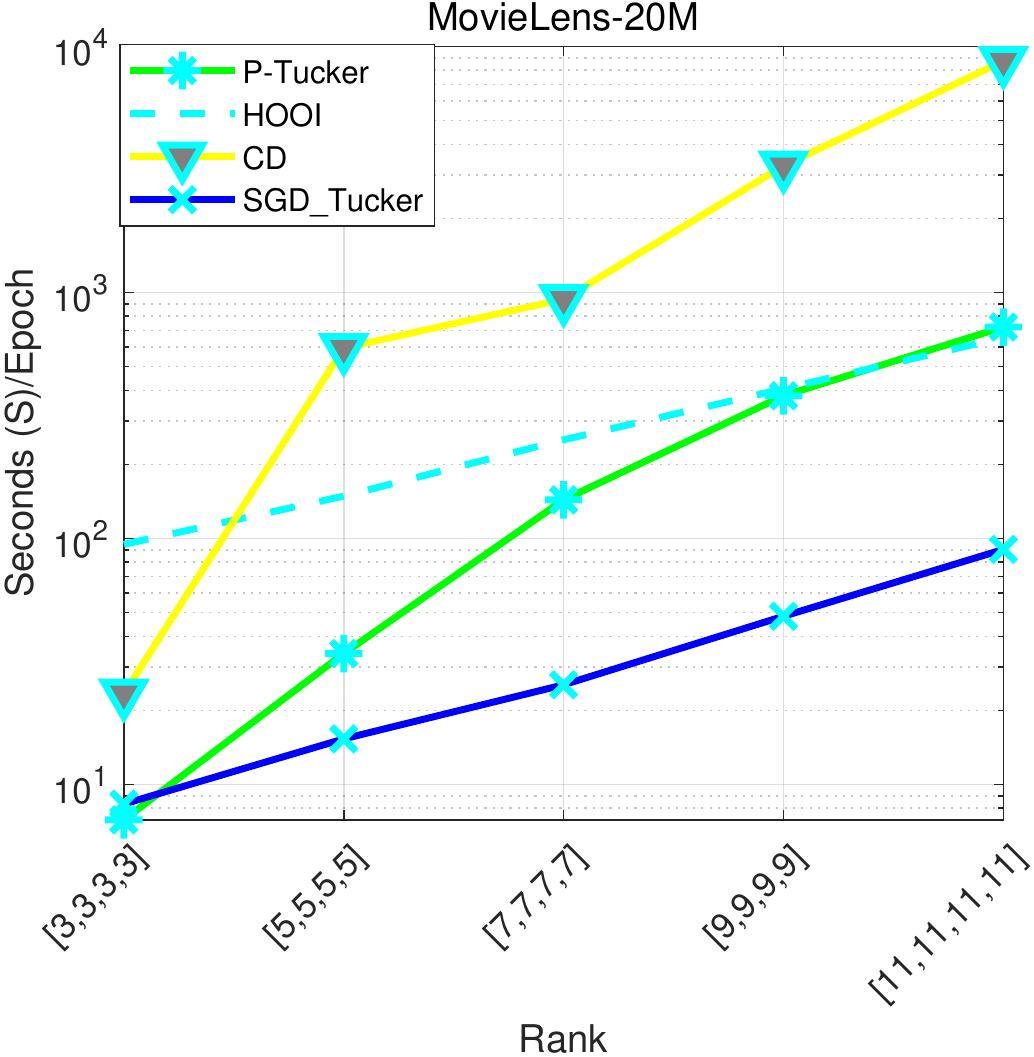}}
    ~
    \subfigure[Netflix-100M]{
    \label{fig401 (a5)} 
    \includegraphics[width=1.4in,height=1.4in]{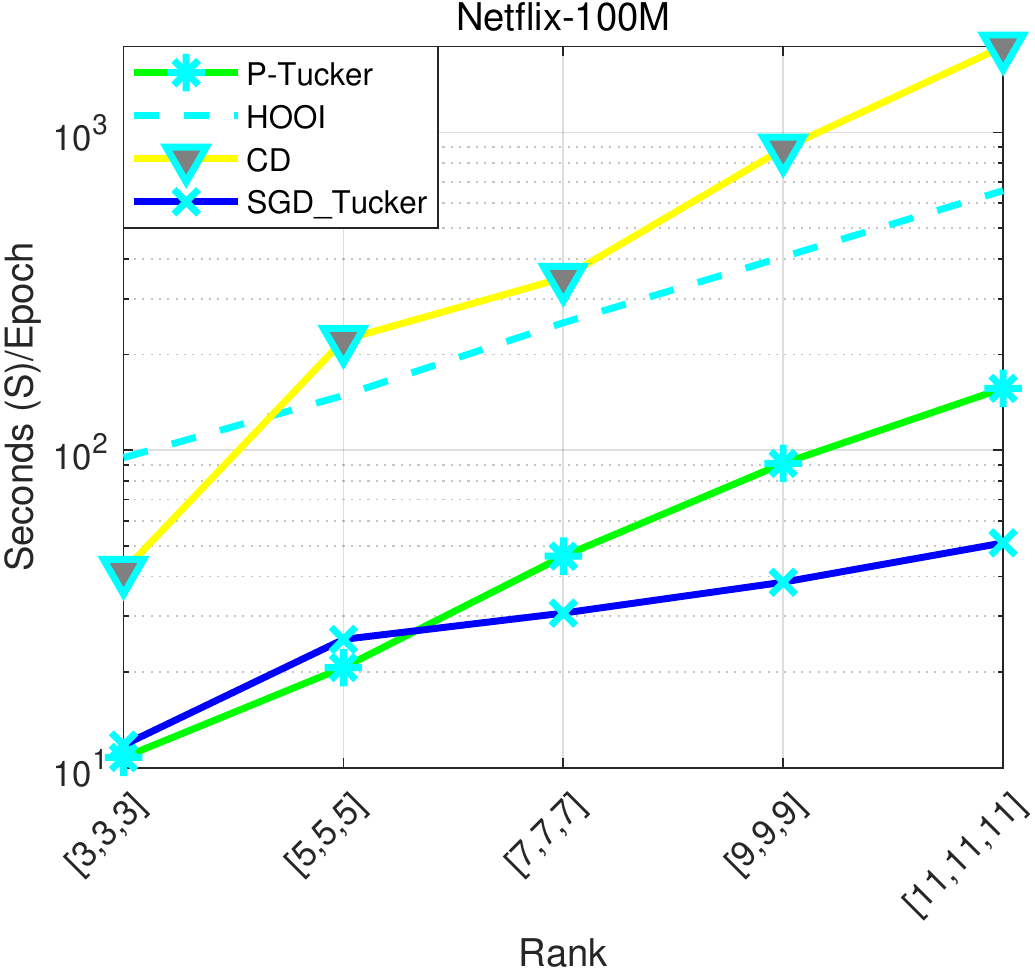}}
        ~
    \subfigure[Yahoo-250M]{
    \label{fig401 (a6)} 
    \includegraphics[width=1.4in,height=1.4in]{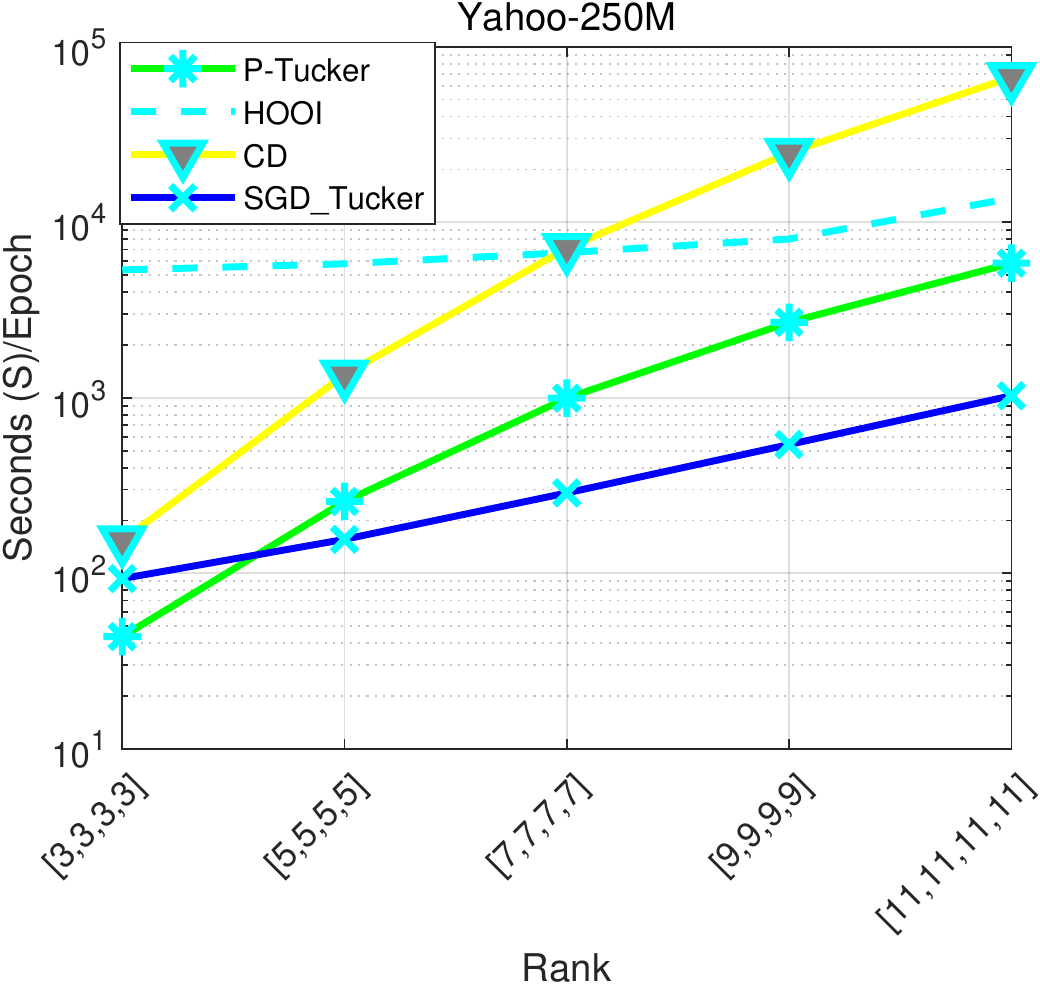}}
\caption{
\textcolor[rgb]{0.00,0.00,1.00}{Rank} scalability for time overhead on full threads.
The computational scalability for P$-$Tucker, HOOI, CD, and SGD$\_$Tucker on
the 4 datasets with successively increased number of total elements, i.e.,Movielen-10M, Movielen-20M, Netflix-100M, and Yahoo-250M.
}
    \label{fig401}
\end{figure}

\begin{figure}[htbp]
  \centering
    \subfigure[Movielen-10M]{
    \label{fig402 (a3)} 
    \includegraphics[width=1.4in,height=1.4in]{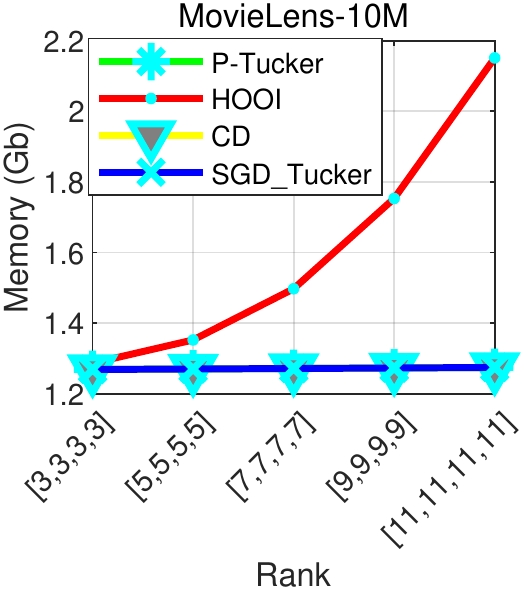}}
    ~
    \subfigure[Movielen-20M]{
    \label{fig402 (a4)} 
    \includegraphics[width=1.4in,height=1.4in]{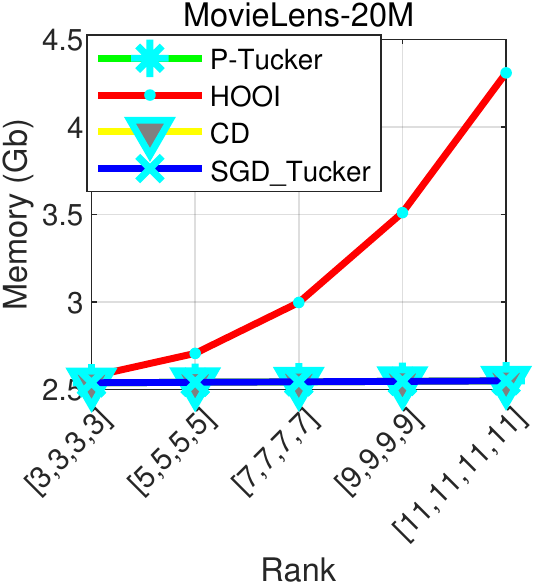}}
    ~
    \subfigure[Netflix-100M]{
    \label{fig402 (a5)} 
    \includegraphics[width=1.4in,height=1.4in]{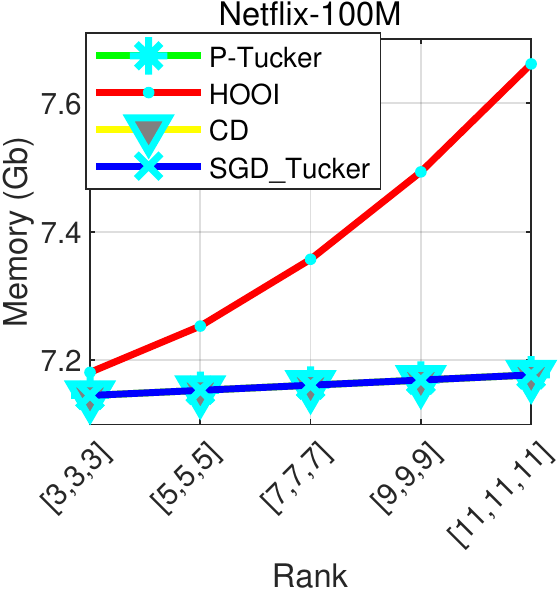}}
        ~
    \subfigure[Yahoo-250M]{
    \label{fig402 (a6)} 
    \includegraphics[width=1.4in,height=1.4in]{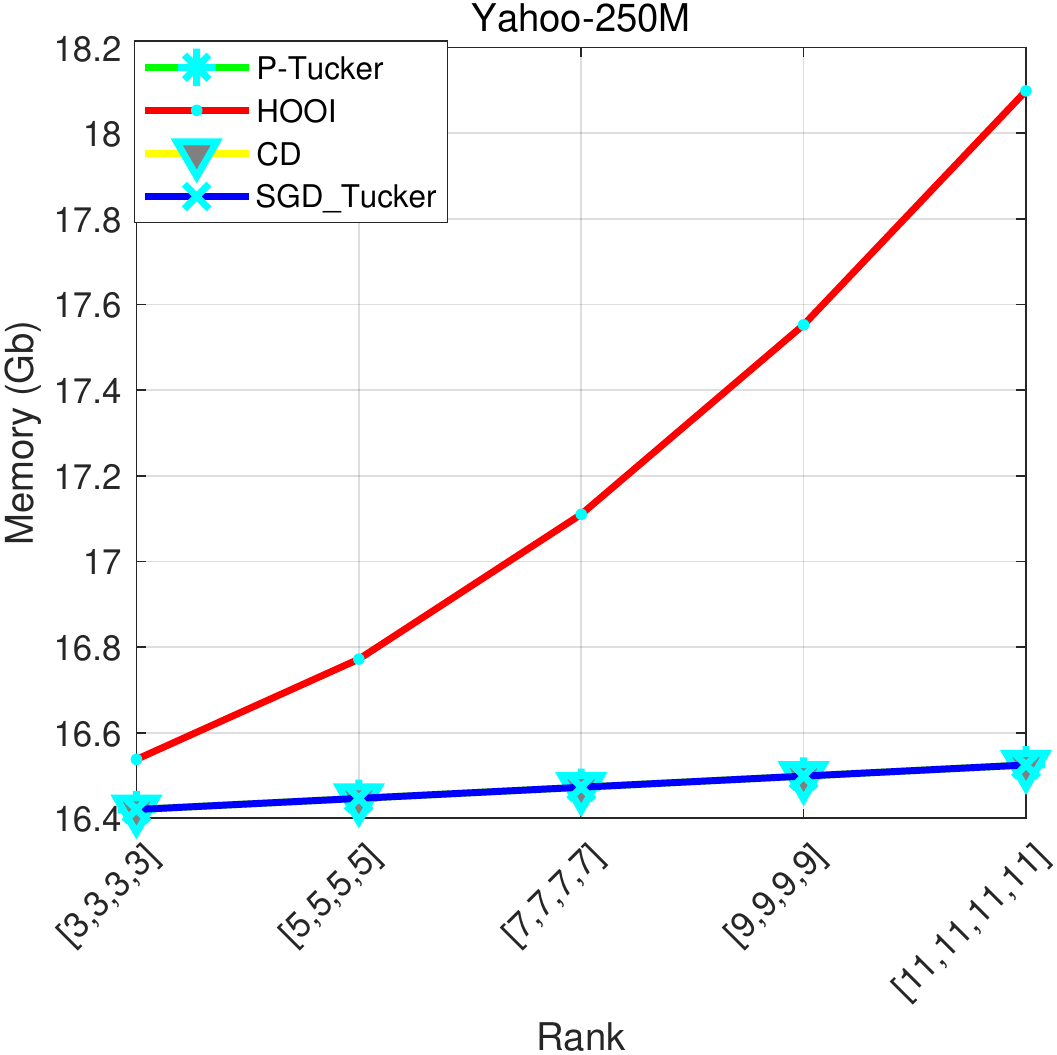}}
         \caption{\textcolor[rgb]{0.00,0.00,1.00}{Rank} scalability for memory overhead on a thread running.
         (In this work,
         \textcolor[rgb]{0.00,0.00,1.00}{GB} refers to GigaBytes and \textcolor[rgb]{0.00,0.00,1.00}{MB} refers to Megabytes).
         The space scalability for P$-$Tucker, CD, HOOI, and SGD$\_$Tucker on
         the 4 datasets with successively increased total elements, i.e., Movielen-10M, Movielen-20M, Netflix-100M, and Yahoo-250M.
}
    \label{fig402}
\end{figure}

\begin{table*}
	\centering
	\footnotesize
	\setlength{\abovecaptionskip}{0pt}
	\caption{Datasets}
	\begin{tabular}{c|cccccc}
		\hline \hline
		\diagbox{Parameters}{Datasets} & Movielens-100K & Movielens-1M & Movielens-10M & Movielens-20M & Netflix-100M & Yahoo-250M \\
        \hline
		$N_{1}$                        & 943        & 6, 040       & 71, 567       & 138, 493      & 480, 189     & 1, 000, 990\\
		$N_{2}$                        & 1, 682           & 3, 706       & 10, 677       & 26, 744       & 17, 770      & 624, 961\\
		$N_{3}$                        & 2              & 4            & 15            & 21            & 2, 182       & 133\\
		$N_{4}$                        & 24             & 24           & 24            & 24            & -            & 24\\
		$|\Omega|$                     & 90, 000        & 990, 252     & 9, 900, 655   & 19, 799, 448  & 99, 072, 112 & 227, 520, 273\\
		$|\Gamma|$                     & 10, 000        &   9, 956     & 99, 398       & 200, 815      & 1, 408, 395  &  25, 280, 002\\
        Order                          &   4            &   4          &     4         &     4         &   3          & 4\\
		Max Value                      & 5.0            & 5.0          & 5.0           & 5.0           & 5.0          & 5.0\\
		Min Value                      & 0.5            & 0.5          & 0.5           & 0.5           & 1.0          & 1.0\\
        \hline
        \hline
	\end{tabular}
	\label{Data sets}
\end{table*}

\section{Evaluation} \label{Section4}
The performance demonstration of SGD$\_$Tucker comprises of two parts:
1) SGD can split the high-dimension intermediate variables into small batches of
intermediate matrices and the scalability and accuracy are presented;
2) SGD$\_$Tucker can be parallelized in a straightfoward manner and
the speedup results are illustrated.
The experimental settings are presented in Section \ref{Section41}.
Sections \ref{Section42} and \ref{Section43}  show the scalability and the influence of parameters, respectively.
At last, Section \ref{Section44} presents the speedup performance of SGD$\_$Tucker and comparison with the state of the art algorithms for STD, i.e., P$-$Tucker \cite{ex234}, CD \cite{ex237}, and HOOI \cite{ex233}.
Due to the limited space,
the full experimental details for HOOI \textcolor[rgb]{0.00,0.00,1.00}{and 2 small datasets, i.e., Movielen-100K and Movielen-1M,} are presented in the supplementary material.

\begin{figure}[htbp]
  \centering
  \subfigure[Computational Time]{
      \label{fig403 (a1)} 
    \includegraphics[width=3.0in,height=2.0in]{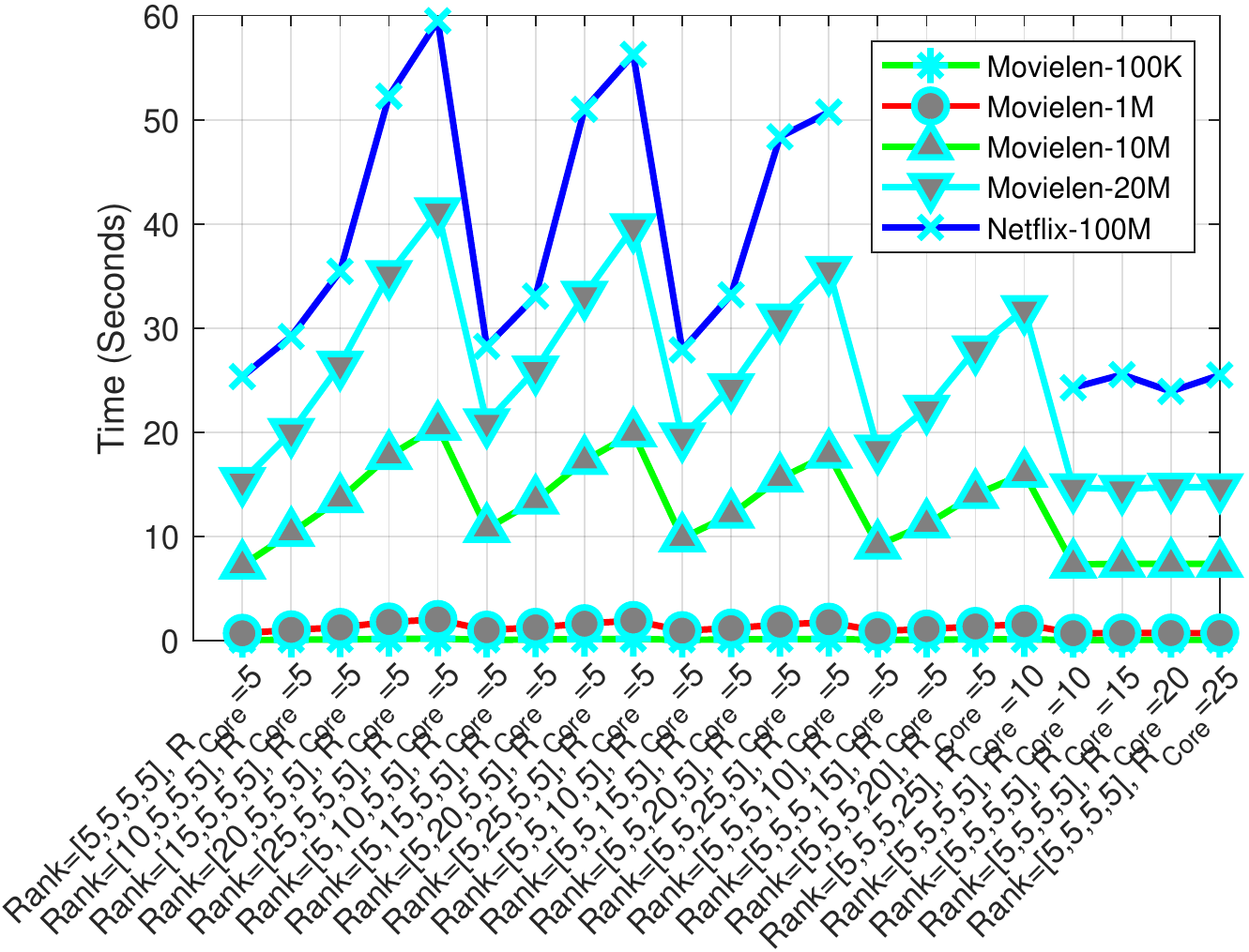}}
    ~
    \subfigure[RMSE]{
    \label{fig403 (a2)} 
    \includegraphics[width=3.0in,height=2.0in]{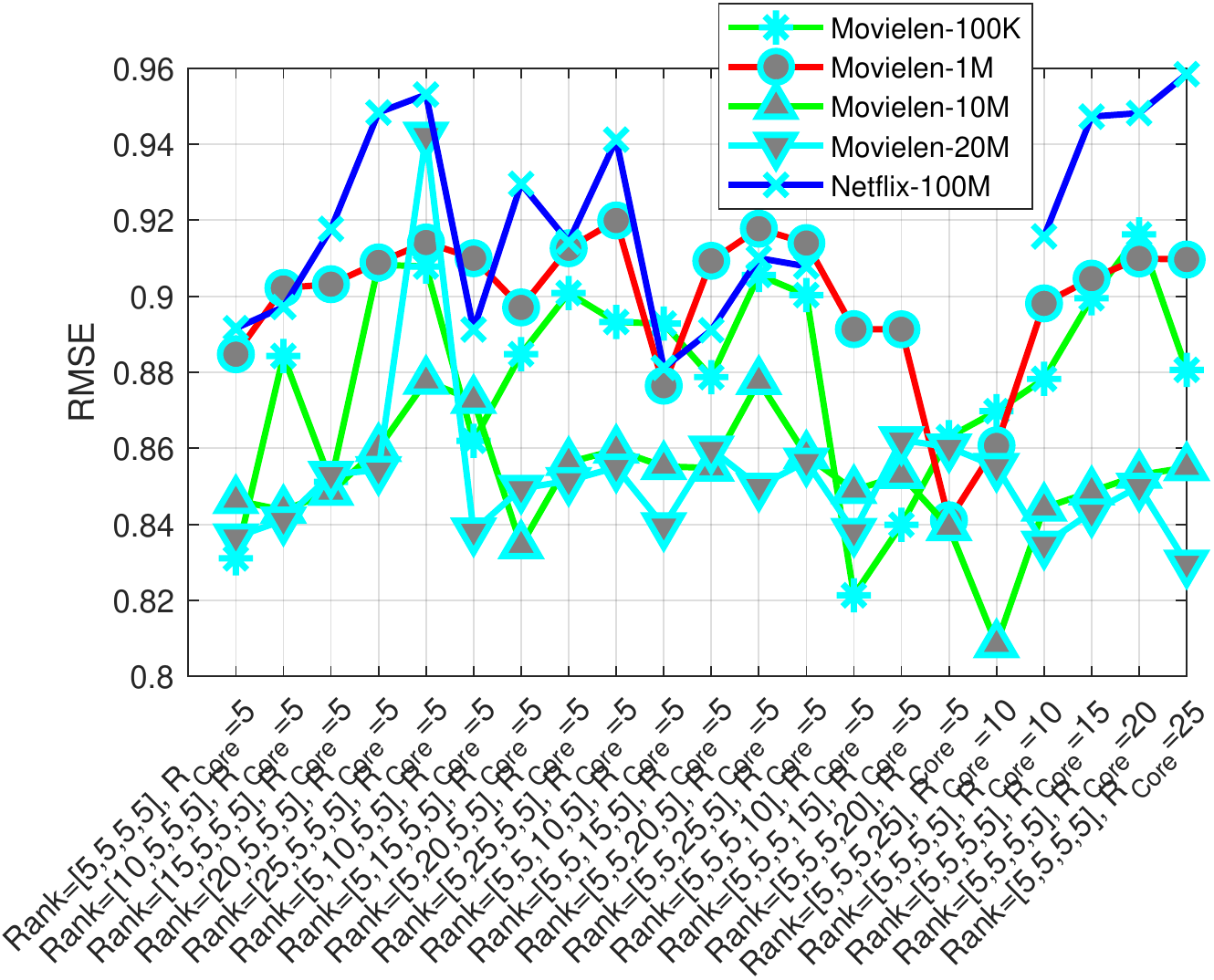}}
    ~
         \caption{Computational overhead, RMSE, and MAE VS. various rank values for SGD$\_$Tucker on 8 Cores.
         Due to limited space, each dataset is presented in $4$-order. As Netflix-100M data set only has $3$-order, its 4th index shall be neglected.}
    \label{fig403}
\end{figure}

\begin{figure}[htbp]
  \centering
    \subfigure[Movielen-10M]{
    \label{fig410 (a3)} 
    \includegraphics[width=1.4in,height=1.4in]{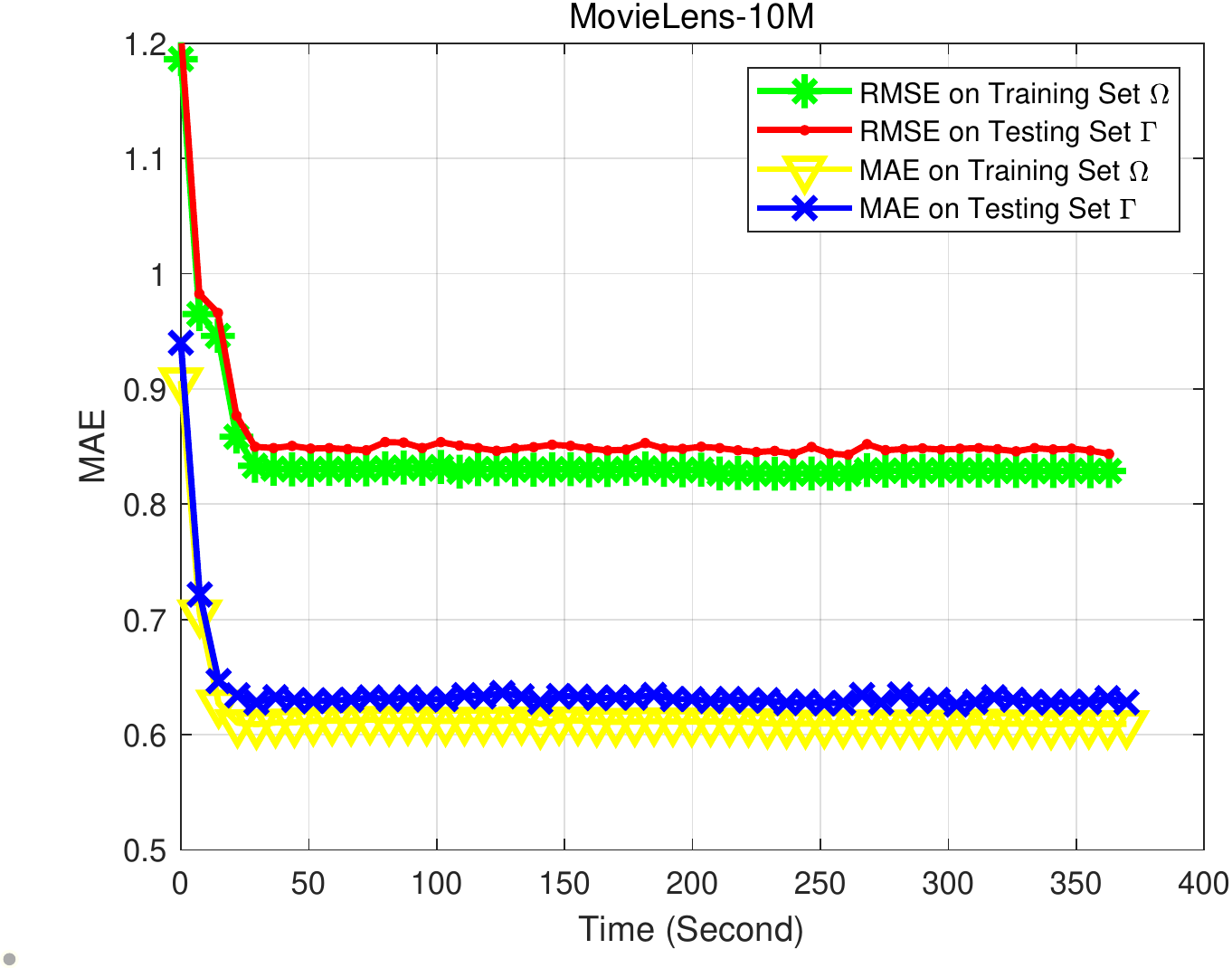}}
    ~
    \subfigure[Movielen-20M]{
    \label{fig410 (a4)} 
    \includegraphics[width=1.4in,height=1.4in]{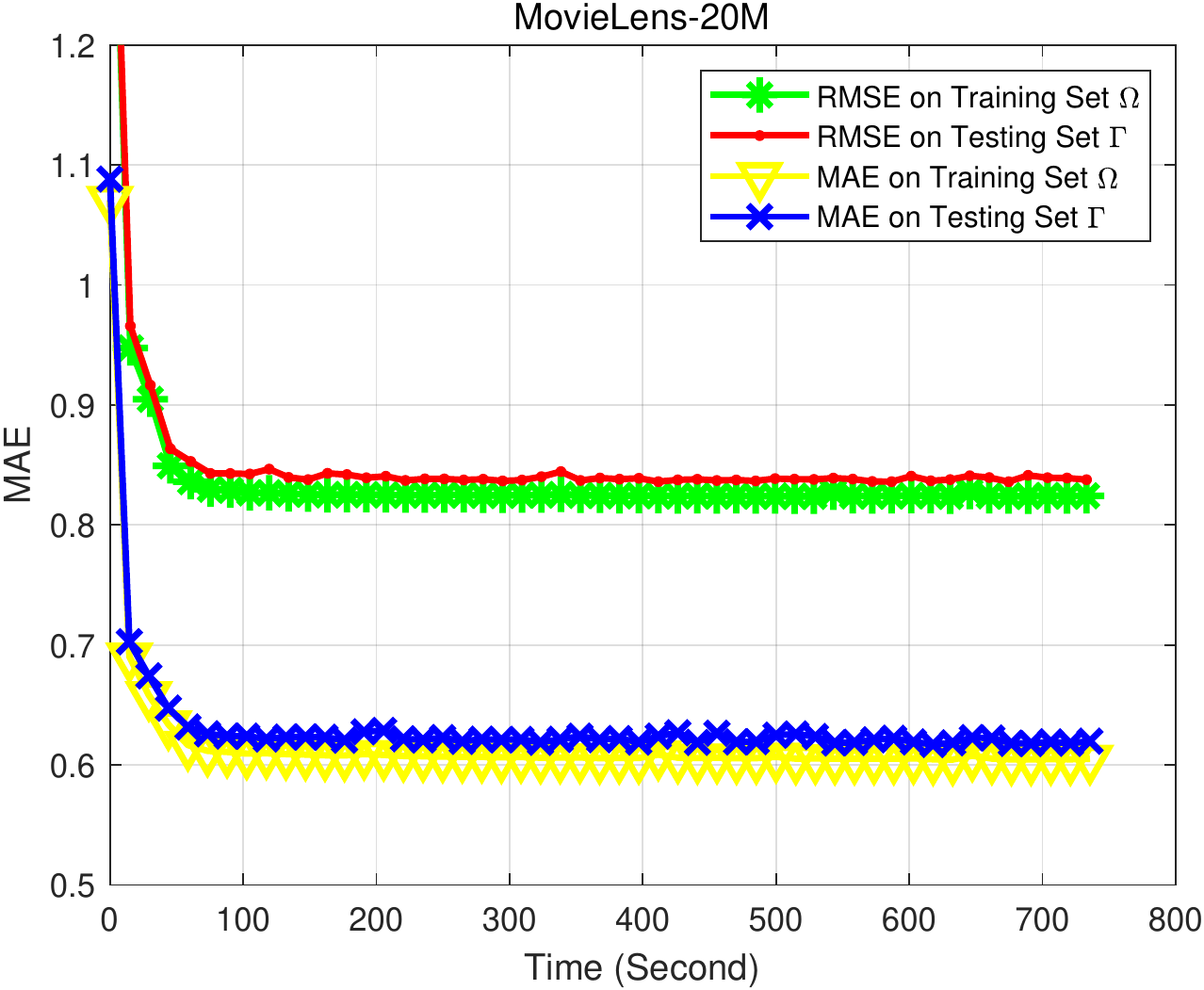}}
    ~
    \subfigure[Netflix-100M]{
    \label{fig410 (a5)} 
    \includegraphics[width=1.4in,height=1.4in]{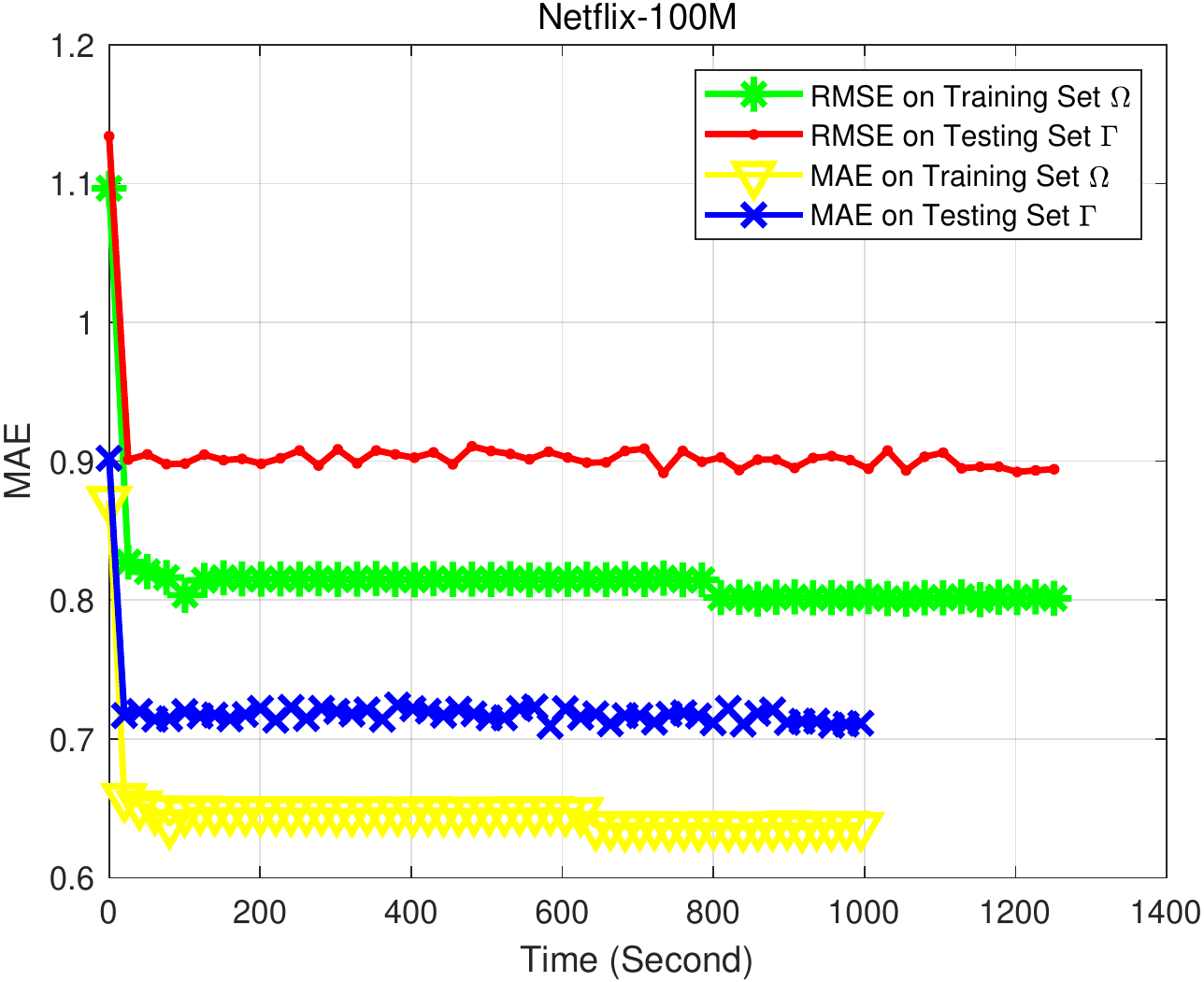}}
        ~
    \subfigure[Yahoo-250M]{
    \label{fig410 (a6)} 
    \includegraphics[width=1.4in,height=1.4in]{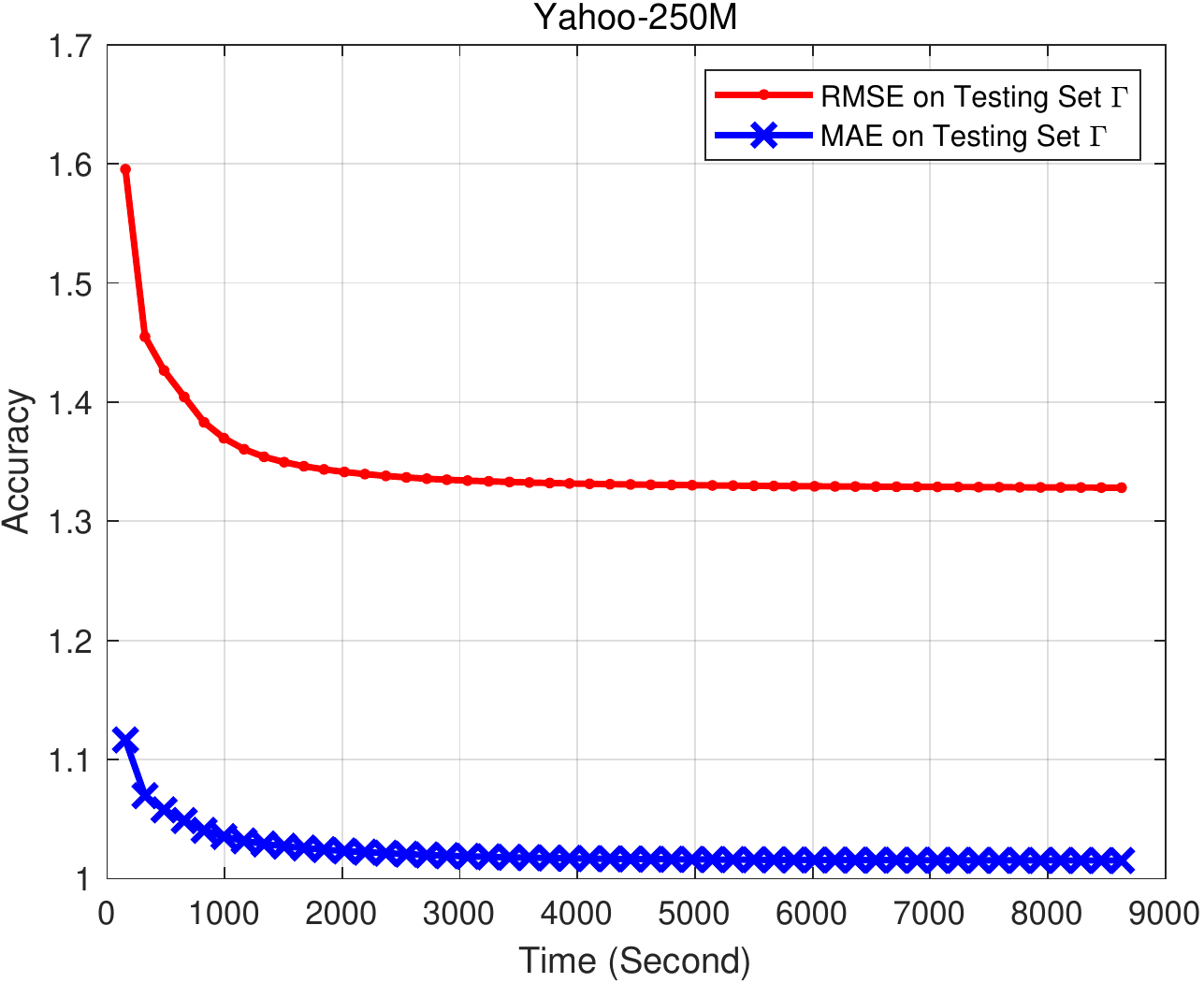}}
         \caption{RMSE and MAE vs time for SGD$\_$Tucker on training set $\Omega$ and testing set $\Gamma$}
    \label{fig410}
\end{figure}

\begin{figure}[htbp]
  \centering
    \subfigure[Movielen-10M]{
    \label{fig405 (a3)} 
    \includegraphics[width=1.4in,height=1.4in]{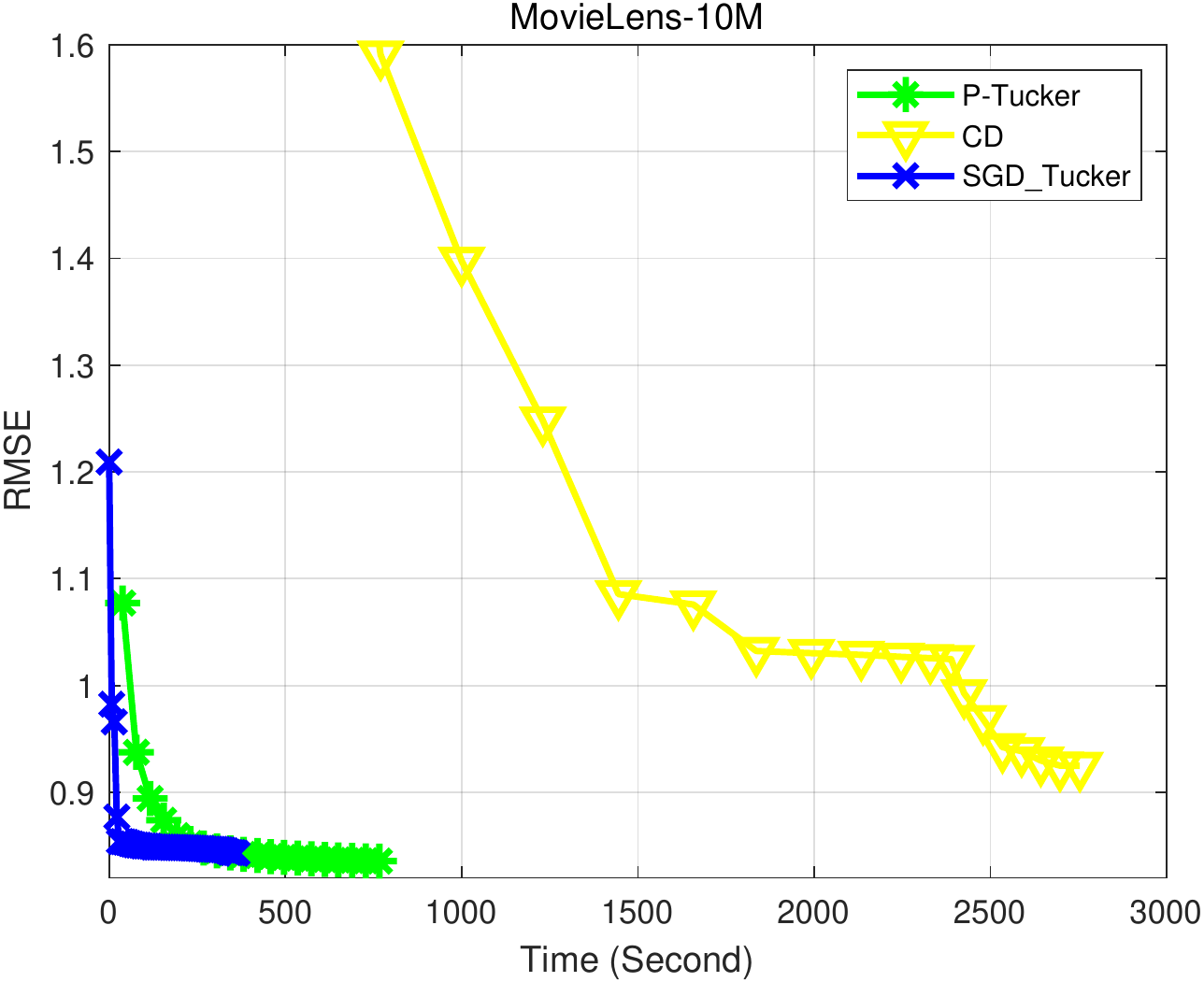}}
    ~
    \subfigure[Movielen-20M]{
    \label{fig405 (a4)} 
    \includegraphics[width=1.4in,height=1.4in]{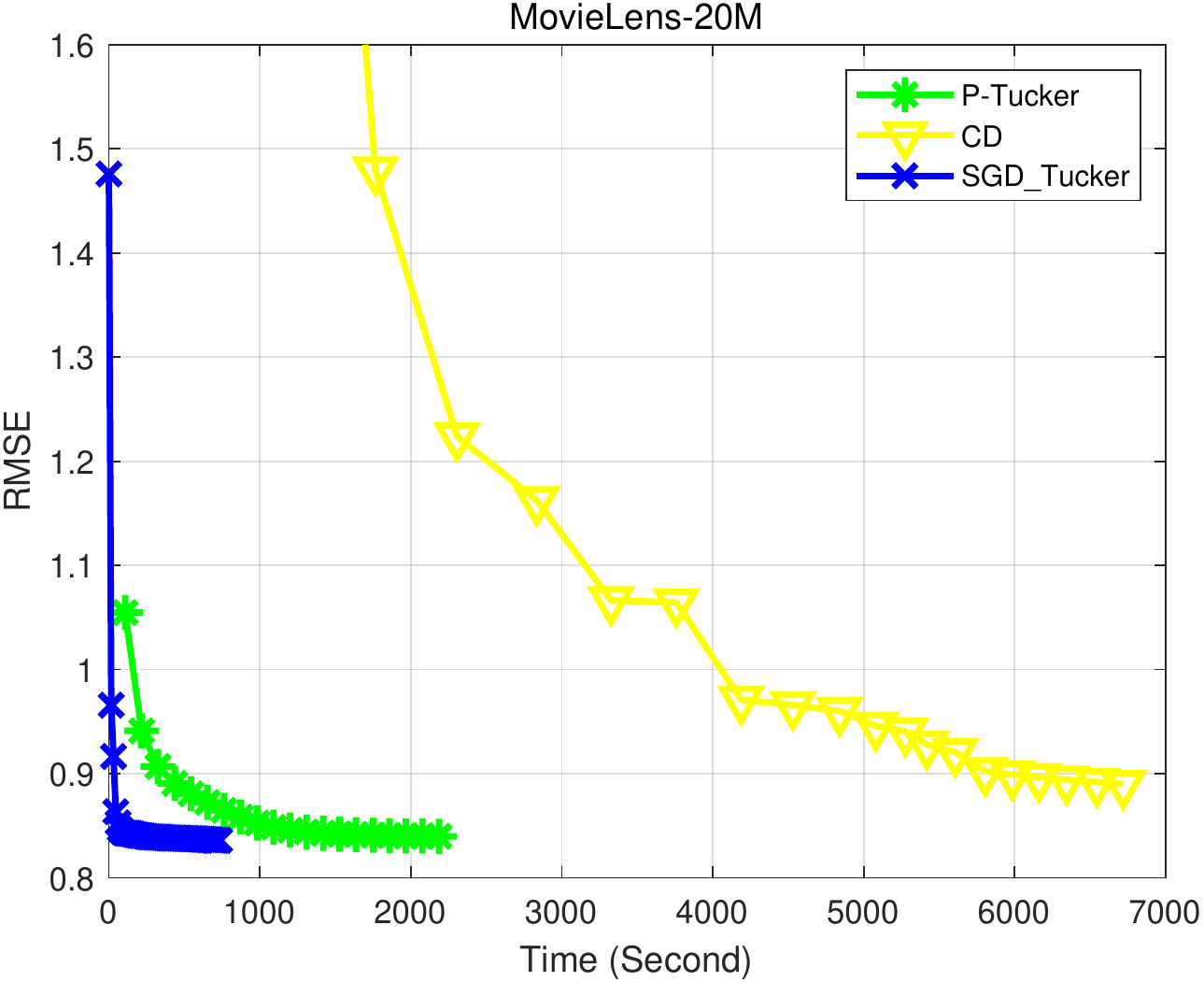}}
    ~
    \subfigure[Netflix-100M]{
    \label{fig405 (a5)} 
    \includegraphics[width=1.4in,height=1.4in]{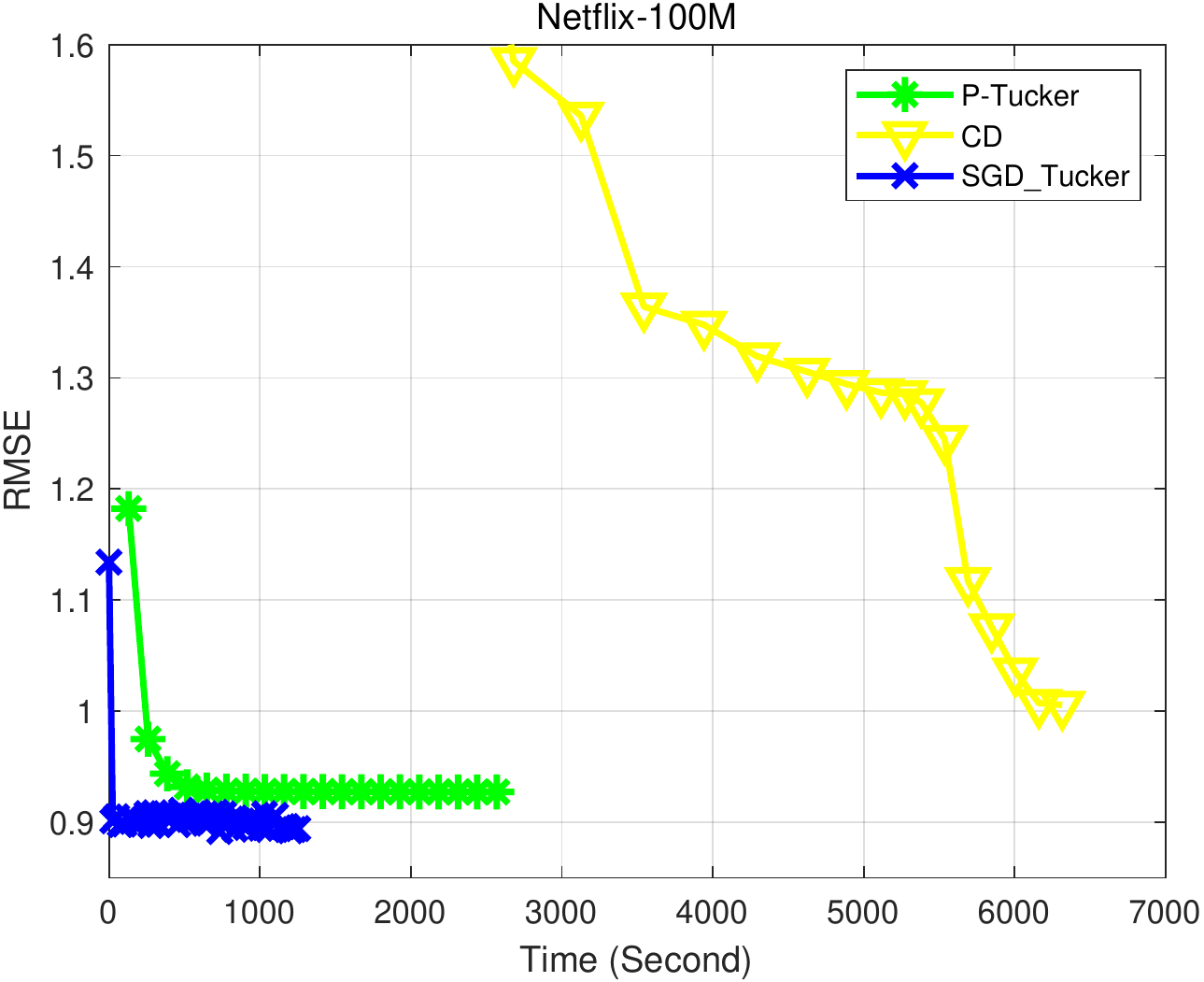}}
        ~
      \subfigure[Yahoo-250M]{
      \label{fig203 (a1)} 
    \includegraphics[width=1.4in,height=1.4in]{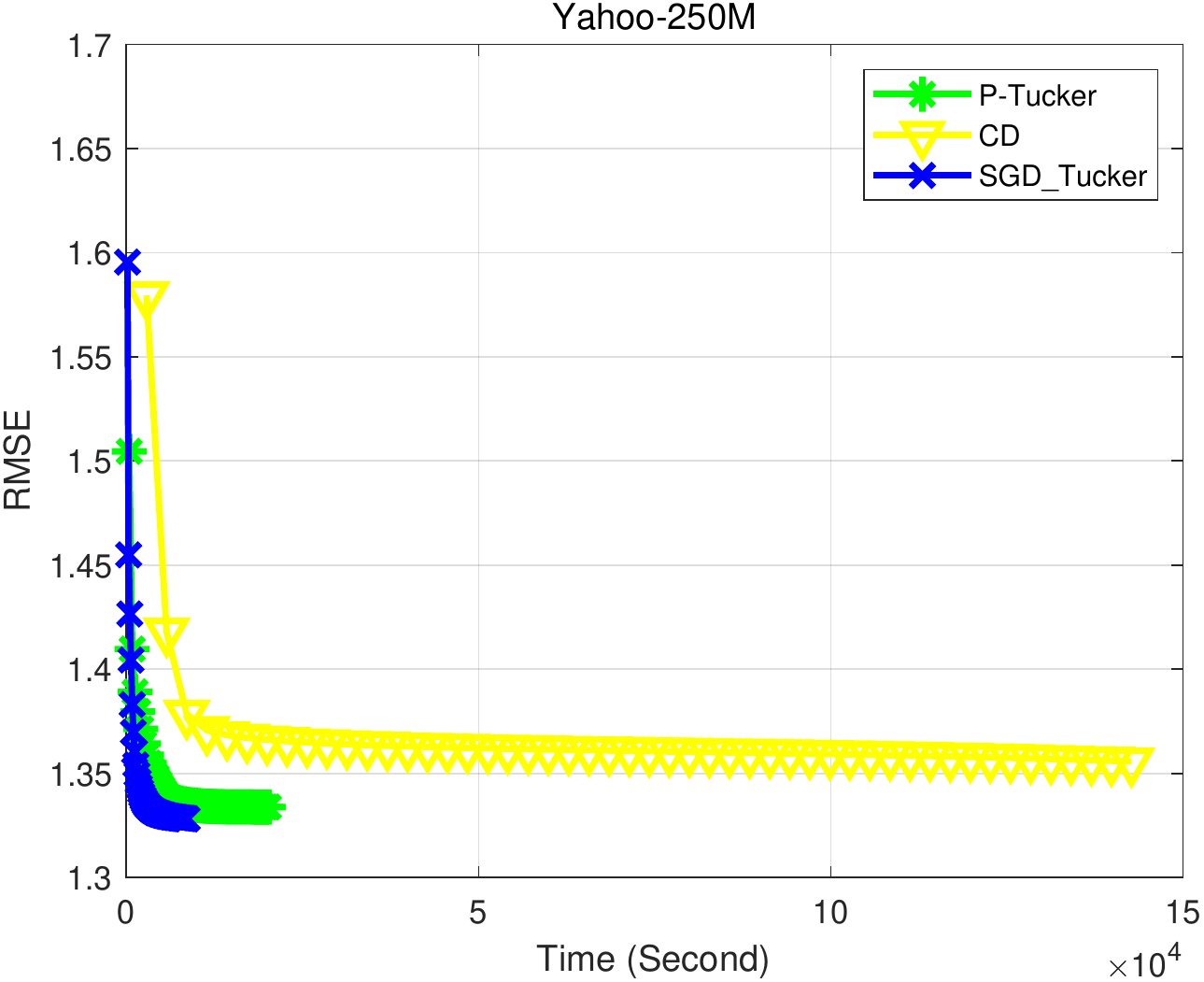}}
         \caption{RMSE comparison of SGD$\_$Tucker, P$-$Tucker, and CD on the 4 datasets.}
    \label{fig405}
\end{figure}

\begin{figure}[htbp]
  \centering
    \includegraphics[width=2.8in,height=2.2in]{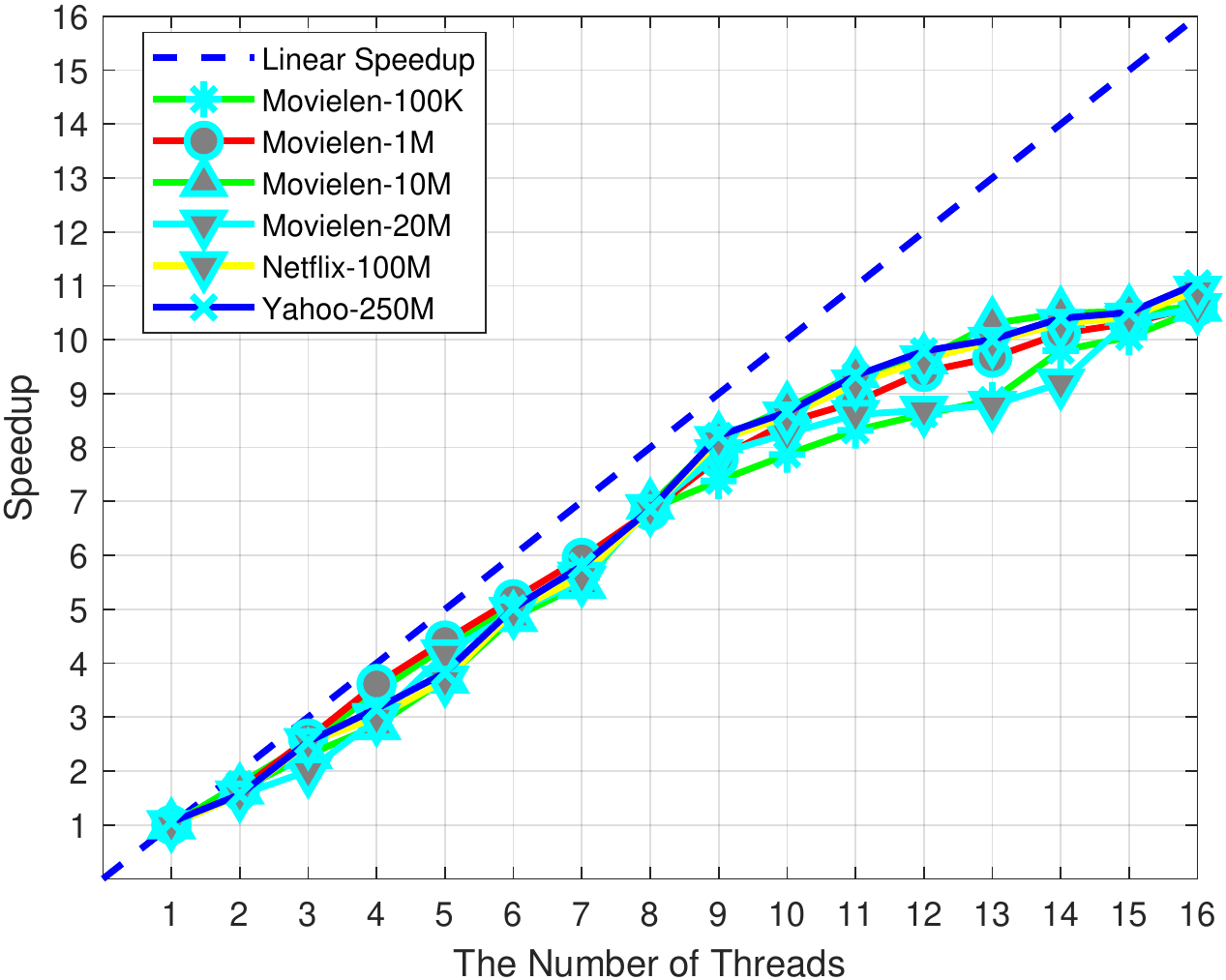}
         \caption{Speedup Performance on the 6 datasets.}
    \label{fig404}
\end{figure}

\subsection{Experimental Settings} \label{Section41}
The CPU server is equipped with 8 Intel(R) Xeon(R) E5-2620 v4 CPUs and each core has 2 hyperthreads,
running on 2.10 GHz, for
the state of the art algorithms for STD, e.g., P$-$Tucker \cite{ex234}, CD \cite{ex237} and HOOI \cite{ex233}.
The experiments are conducted 3 public datasets :
Netflix \footnotemark[1] \footnotetext[1]{https://www.netflixprize.com/},
Movielens \footnotemark[2] \footnotetext[2]{https://grouplens.org/datasets/movielens/}, and
Yahoo-music \footnotemark[3] \footnotetext[3]{https://webscope.sandbox.yahoo.com/catalog.php?datatype}.
\textcolor[rgb]{0.00,0.00,1.00}{The datasets which be used in our experiments can be downloaded in this link.} \footnotemark[4] \footnotetext[4] {https://drive.google.com/drive/folders/1tcnAsUSC9jFty7AfWetb5R\\7nevN2WUxH?usp=sharing}
For Movielens, data cleaning is conducted and $0$ values are changed from zero to $0.5$ and
we make compression for Yahoo-music dataset in which the maximum value 100 be compressed to 5.0.
The specific situation of datasets is shown in Table \ref{Data sets}.
In this way, the ratings of the two data sets are in the same interval,
which facilitates the selection of parameters.
Gaussian distribution $N\bigg(0.5,0.1^{2}\bigg)$ is adopted for initialization;
meanwhile, the regularization parameters $\{\lambda_{\textbf{A}}, \lambda_{\textbf{B}}\}$ are set as $0.01$ and
the learning rate $\{\gamma_{\textbf{A}}, \gamma_{\textbf{B}}\}$ are set as $\{0.002, 0.001\}$, respectively.
For simplification, we set $M=1$.
The accuracy is measured by $RMSE$ and $MAE$ as:
\begin{equation}
\begin{aligned}
\setlength{\abovedisplayskip}{0pt}
\setlength{\belowdisplayskip}{0pt}
RMSE&=\sqrt{\bigg(\sum_{(i_{1},\cdots,i_{N})\in\Gamma}(x_{i_{1},\cdots,i_{N}}-\widehat{x}_{i_{1},\cdots,i_{N}})^{2}\bigg)\bigg/|\Gamma|};\\
MAE&=\bigg(\sum_{(i_{1},\cdots,i_{N})\in\Gamma}\bigg|x_{i_{1},\cdots,i_{N}}-\widehat{x}_{i_{1},\cdots,i_{N}}\bigg|\bigg)\bigg/|\Gamma|,\\
\end{aligned}
\end{equation}
respectively,
where $\Gamma$ denotes the test sets.
All the experiments are conducted on double-precision float numeric.

\subsection{Scalability of the Approach} \label{Section42}
When ML algorithms process large-scale data,
two factors to influence the computational overhead are:
1) time spent for data access on memory, i.e., reading and writing data;
2) time spent on the computational components for executing the ML algorithms.
Fig. \ref{fig401} presents the computational overhead per training epoch.
The codes of P$-$Tucker and CD have fixed settings for the rank value, i.e., $J_{1}=\cdots=J_{n}=\cdots=J_{N}$.
The rank value is set as an increasing order, $\{3, 5, 7, 9, 11\}$.
As shown in Fig. \ref{fig401}, SGD$\_$Tucker has the lowest computational time overhead.
HOOI needs to construct the intermediate matrices $\textbf{Y}_{(n)}$ $\in$ $\mathbb{R}^{I_{n}\times \prod\limits_{k\neq n}^{N}J_{k}}$, $n$ $\in$ $\{N\}$ and SVD for $\textbf{Y}_{(n)}$ (In supplementary material).
P$-$Tucker should construct Hessian matrix and the computational complexity of each Hessian matrix inversion is  $O\big(J_{n}^{3}\big)$, $n$ $\in$ $\{N\}$.
P$-$Tucker saves the memory overhead a lot at the expense of adding the memory accessing overhead.
Eventually, the accessing and computational overheads for inverse operations of Hessian matrices make the overall computational time of P$-$Tucker surpass SGD$\_$Tucker.
The computational structure of CD has linear scalability.
However, the CD makes the update process of each feature element in a feature vector be discrete.
Thus, the addressing process is time-consuming.
Meanwhile, as Fig. \ref{fig401} shows, the computational overhead satisfies the constraints of the time complexity analyses, i.e.,
P$-$Tucker, HOOI, CD (Supplementary material)
and SGD$\_$Tucker (Section \ref{Section34}).

The main space overheads for STD comes from the storage need of the intermediate matrices.
The SVD operation for the intermediate matrices $\textbf{Y}_{(n)}$ $\in$ $\mathbb{R}^{I_{n}\times \prod\limits_{k\neq n}^{N}J_{k}}$, $n$ $\in$ $\{N\}$ makes the HOOI scale
exponentially.
P$-$Tucker constructs the eventual Hessian matrices directly which can avoid the huge construction of intermediate matrices.
However, the overhead for addressing is huge and the long time-consuming for CD lies in the same situation.
CD is a linear approach from the update equation (Supplementary material).
However, in reality, the linearity comes from  the approximation of the second-order Hessian matrices and, thus, the accessing overhead of discrete elements of CD overwhelms the accessing overhead of continuous ones (P$-$Tucker and SGD$\_$Tucker).
As shown in Fig. \ref{fig402},
SGD$\_$Tucker has linear scalability for space overhead and
HOOI is exponentially scalable.
P$-$Tucker and CD have the near-linear space overhead.
However, as Fig. \ref{fig401} shows, the computational overheads of P$-$Tucker and CD is significantly higher than SGD$\_$Tucker.
Meanwhile, as Fig. \ref{fig402} shows, the space overhead satisfies the constraints of the space complexity analyses, i.e.,
HOOI, P$-$Tucker, and CD (Supplementary material) and SGD$\_$Tucker (Section \ref{Section34}).

\subsection{The Influence of Parameters} \label{Section43}
In this section, the influence of the rank on the behavior of the algorithm is presented.
We summarize the influence on the computational time and the RMSE are presented in Fig. \ref{fig403}.
The Netflix-100M dataset has only 3-order.
Owing to the limited space,
the performances on Netflix-100M dataset are combined with other 5 datasets
(Netflix-100M dataset has only 3-order and the performances on the 4th order of Netflix-100M dataset should be negligible).
The update for the Kruskal matrices $\textbf{B}^{(n)}$, $n$ $\in$ $N$ on the steps of beginning and last epochs can also obtain a
comparable results.
Thus, the presentation for the Kruskal matrices $\textbf{B}^{(n)}$, $n$ $\in$ $N$ is omitted.
The influence for computational time is divided into 5 sets, i.e.,
$\big\{J_{1}\in\{5, 10, 15, 20, 25\} , J_{k}=5, k\neq 1, R_{Core}=5\big\}$,
$\big\{J_{2}\in\{5, 10, 15, 20, 25\} , J_{k}=5, k\neq 2, R_{Core}=5\big\}$,
$\big\{J_{3}\in\{5, 10, 15, 20, 25\} , J_{k}=5, k\neq 3, R_{Core}=5\big\}$,
$\big\{J_{4}\in\{5, 10, 15, 20, 25\} , J_{k}=5, k\neq 4, R_{Core}=5\big\}$,
$\big\{J_{n}=5, n\in\{N\}, R_{Core}\in\{10,15,20,25\}\big\}$.
As the Fig. \ref{fig403} (a) shows,
the computation time increases with $J_{n}$, $n$ $\in$ $\{N\}$ increasing linearly and
the $R_{Core}$ only has slight influence for computational time.

The codes of P$-$Tucker and CD have fixed settings for the rank value ($J_{1}=\cdots=J_{n}=\cdots=J_{N}$).
For a fair comparison, the rank values that we select should have slight changes with the RMSE performances of other rank values.
The computational overhead of P$-$Tucker and CD is sensitive to the rank value and RMSE is comparable non-sensitive.
Hence, we choose a slightly small value $J_{n}=5| n\in\{N\}$.
As Fig. \ref{fig403} (b) shows, when rank is set as $J_{n}=5, n\in\{N\}$, the RMSE is equivalent to other situations on average.

\subsection{Speedup and Comparisons} \label{Section44}
We test the RMSE performances on the 6 datasets and the RMSE is used to estimate the missing entries.
The speedup is evaluated as $Time_{1}/Time_{T}$ where $Time_{T}$ is the running time with $T$ threads.
When $T$ $=$ $\{1,\cdots,8\}$,
as the Fig. \ref{fig404} shows, the speedup performance of SGD$\_$Tucker has a near-linear speedup.
The speedup performance is presented in Fig. \ref{fig404} which shows the whole speedup performance of the SGD$\_$Tucker.
The speedup performance is a bit more slower when the number of threads is larger than $8$.
The reason is that the thread scheduling and synchronization occupies a large part of time.
When we use $16$ threads, there is still a $11X$ speedup and the efficiency is $11/16=68\%$.

The rank value for $3$-order tensor is set to $[5,5,5]$ and the rank value for $4$-order tensor is set to $[5,5,5,5]$.
To demonstrate the convergence of SGD$\_$Tucker, the convergence performances of SGD$\_$Tucker for the Movielen and Netflix are presented in Fig. \ref{fig410}.
As shown in Fig. \ref{fig410}, SGD$\_$Tucker can get more stable RMSE and MAE metrices which
mean that SGD$\_$Tucker has more robustness in large-scale datasets than in small datasets.
As Fig. \ref{fig405} shows, SGD$\_$Tucker can not only run faster than the state of the art approaches,
i.e., P$-$Tucker and CD,
but also can obtain higher RMSE value.
SGE$\_$Tucker runs 2$X$ and 20$X$ faster than P$-$Tucker and CD, respectively, to obtain the optimal RMSE value.

\section{Conclusion and Future Works} \label{Section5}
STD is widely used in low-rank representation learning for sparse big data analysis.
Due to the entanglement problem of core tensor and factor matrices,
the computational process for STD has the intermediate variables explosion problem due to following all elements in HOHDST.
To solve this problem, we first derive novel optimization objection function for STD and then propose SGD$\_$Tucker to solve it by dividing the high-dimension intermediate variables into small batches of intermediate matrices. 
Meanwhile, the low data-dependence of SGD$\_$Tucker makes it amenable to
fine-grained parallelization.
The experimental results show that SGD$\_$Tucker has linear computational time and space overheads and
SGD$\_$Tucker runs at least 2$X$ faster than the state of the art STD solutions.
In the future works,
we plan to explore how to accelerate the SGD$\_$Tucker by
the state of the art stochastic models, e.g.,
variance SGD \cite{ex225},
Stochastic Recursive Gradient \cite{ex230}, or
momentum SGD \cite{ex226}.

SGD$\_$Tucker is a linearly scalable method for STD on big-data platforms.
For the future work, we will embed SGD$\_$Tucker into popular distributed data processing platforms such as Hadoop, Spark, or Flink.
We will also aim to extend SGD$\_$Tucker to GPU platforms, i.e., OpenCL and CUDA.

\section*{Acknowledgement}

\ifCLASSOPTIONcaptionsoff
  \newpage
\fi
\small
\bibliographystyle{IEEEtran}
\bibliography{bib}

\end{document}